\documentclass[review, 1p]{elsarticle}
\usepackage{blkarray}
\usepackage[most]{tcolorbox}
\usepackage{tikz,lipsum,lmodern}
\usepackage{wrapfig}
\usepackage{cleveref}
\usepackage{amssymb}
\usepackage{mathtools}
\usepackage{algorithm}
\usepackage{algpseudocode}
\usepackage{amsthm}
\usepackage{lmodern}
\usepackage{relsize}
\usepackage{makeidx}
\usepackage{bold-extra}
\usepackage{graphicx}
\usepackage[T1]{fontenc}
\usepackage{lmodern}
\usepackage{textcomp}
\usepackage[OT1]{eulervm}
\usepackage{lineno,hyperref}
\modulolinenumbers[200]
\usepackage[T1]{fontenc}
\usepackage{lmodern}
\usepackage{textcomp}
\usepackage{amsfonts}
\usepackage{amsmath}
\usepackage{color,soul}
\usepackage{amssymb}
\usepackage{lipsum}
\usepackage{hyperref}
\newtheorem{theorem}{Theorem}
\newtheorem{lemma}{Lemma}
\newtheorem{corollary}{Corollary}

\renewcommand{\vec}[1]{\mathbf{#1}}
\newtheorem{Observation}{Observation}
\newtheorem{definition}{Definition}
\newtheorem{claim}{Claim}
\journal{Elsevier}
\hypersetup{
    colorlinks=true,
    linkcolor=blue,
    filecolor=magenta,      
    urlcolor=cyan,
}
\tcbset
    {
    		enhanced,
        left=3mm,
        right=3mm,
        width=135mm, 
        boxrule=0.4pt,
        colback=red!5!white,
        boxrule=0.1pt,
        colframe=red!75!black,fonttitle=\bfseries,
    }

    \makeatletter
    \renewcommand{\paragraph}{\@startsection{paragraph}{4}{\z@}%
       {-3.25ex\@plus -1ex \@minus -.2ex}%
       {1.5ex \@plus .2ex}%
       {\normalfont\normalsize\bfseries}}
    \renewcommand{\subparagraph}{\@startsection{subparagraph}{5}{\z@}%
       {-3.25ex\@plus -1ex \@minus -.2ex}%
       {1.5ex \@plus .2ex}%
       {\normalfont\normalsize\bfseries}}
    \makeatother
\begin{document}
\newsavebox{\mybox}
\savebox{\mybox}{This is a tcolorbox}
\begin{frontmatter}

\title{Simultaneous Consecutive Ones Submatrix and Editing Problems : Classical Complexity \& Fixed-Parameter Tractable Results }

\author{Rani M. R, R. Subashini, Mohith Jagalmohanan}
\address{Department of Computer Science \& Engineering.\\ National Institute of Technology, Calicut\\
 \texttt{\{rani\_p150067cs, suba\}@nitc.ac.in}}

\begin{abstract}
A binary matrix $M$ has the consecutive ones property ($C1P$) for rows (resp. columns) if there is a permutation of its columns (resp. rows) that arranges the ones consecutively in all the rows (resp. columns). If $M$ has the $C1P$ for rows and the $C1P$ for columns, then $M$ is said to have the simultaneous consecutive ones property ($SC1P$). Binary matrices having the $SC1P$ plays an important role in theoretical as well as practical applications. \newline \newline In this article, we consider the classical complexity and fixed-parameter tractability of $(a)$ Simultaneous Consecutive Ones Submatrix ($SC1S$) and $(b)$ Simultaneous Consecutive Ones Editing ($SC1E$) [Oswald et al., Theoretical Comp. Sci. 410(21-23):1986-1992, \hyperref[references]{2009}] problems. $SC1S$ problems focus on deleting a minimum number of rows, columns, and rows as well as columns to establish the $SC1P$, whereas $SC1E$ problems deal with flipping a minimum number of $0$-entries, $1$-entries, and $0$-entries as well as $1$-entries to obtain the $SC1P$.  We show that the decision versions of $SC1S$ and $SC1E$ problems are NP-complete. \newline \newline We consider the parameterized versions of $SC1S$ and $SC1E$ problems with $d$, being the solution size, as the parameter. Given a binary matrix $M$ and a positive integer $d$, $d$-$SC1S$-$R$, $d$-$SC1S$-$C$, and  $d$-$SC1S$-$RC$ problems decide whether there exists a set of rows, columns, and rows as well as columns, respectively, of size at most $d$, whose deletion results in a matrix with the $SC1P$. The $d$-$SC1P$-$0E$, $d$-$SC1P$-$1E$, and $d$-$SC1P$-$01E$ problems decide whether there exists a set of $0$-entries, $1$-entries, and $0$-entries as well as $1$-entries, respectively, of size at most $d$, whose flipping results in a matrix with the \vspace{0.095 in} $SC1P$. 

Our main results include: \begin{enumerate} \item The decision versions of $SC1S$ and $SC1E$ problems are NP-complete. \item Using bounded search tree technique, certain reductions and related results from the literature [Cao et al., Algorithmica 75(1):118-137, \hyperref[references]{2016}, and Kaplan et al., SIAM Journal on Computing 28(5):1906-1922, \hyperref[references]{1999}], we show that $d$-$SC1S$-$R$, $d$-$SC1S$-$C$, $d$-$SC1S$-$RC$ and $d$-$SC1P$-$0E$ are fixed-parameter tractable on binary matrices with run-times $O^{*}(8^{d})$, $O^{*}(8^{d})$, $O^{*}(2^{O(dlogd)})$ and $O^{*}(18^{d})$ respectively. \end{enumerate}We also give improved FPT algorithms for $SC1S$ and $SC1E$ problems on certain restricted binary matrices. 
\end{abstract}
\end{frontmatter}
\linenumbers
\section{Introduction}\label{intro}
Binary matrices having the simultaneous consecutive ones property are fundamental in recognizing biconvex graphs \cite{tucker1972structure}, recognizing proper interval graphs \cite{Fis85}, identifying block structure of matrices in applications arising from integer linear programming \cite{oswald2003weighted}  and finding clusters of ones from metabolic networks \cite{konig2006discovering}. A binary matrix  has the \textit{consecutive ones property ($C1P$) for rows (resp. columns)} \cite{fulkerson1965incidence}, if there is a permutation of its columns (resp. rows) that arranges the ones consecutively in all the rows (resp. columns). A binary  matrix has the \textit{simultaneous consecutive ones property} \textbf{\textit{($SC1P$)}}  \cite{oswald2009simultaneous}, if we can permute the rows and columns in such a way that the ones in every column and in every row occur consecutively. That is, a binary matrix has the $SC1P$ if it satisfies the $C1P$ for both rows and columns. Matrices with the $C1P$ and the $SC1P$ are related to interval graphs and proper interval graphs respectively. There exist several linear-time and polynomial-time algorithms for testing the $C1P$ for columns (see, for example \cite{booth1976testing,hsu2002simple,hsu2003pc,mcconnell2004certifying,meidanis1998consecutive,raffinot2011consecutive}). 
These algorithms can also be used for testing the $C1P$ for rows. The column permutation (if one exists) to obtain the $C1P$ for rows will not affect the $C1P$ of the columns (if one exists) and vice versa. Thus, testing the $SC1P$ can also be done in linear time. \\\\$SC1P$ being a non-trivial property, we aim to establish the $SC1P$ in a given binary matrix through deletion of row(s)/column(s) and flipping of $0$s/$1$s. We consider the Simultaneous Consecutive Ones Submatrix ($SC1S$) and Simultaneous Consecutive Ones Editing ($SC1E$) \cite{oswald2009simultaneous} problems to establish the $SC1P$, if the given binary matrix do not have the $SC1P$. \textit{SC1S problems} focus on deleting a minimum number of rows, columns, and rows as well as columns to establish the $SC1P$ whereas \textit{SC1E problems} deal with flipping a minimum number of $0$-entries, $1$-entries, and $0$-entries as well as $1$-entries to obtain the $SC1P$. We pose the following optimization problems: \textsc{Sc$\mathsmaller1$s-Row Deletion}\label{opt}, \textsc{Sc$\mathsmaller 1$s-Column Deletion} and \textsc{Sc$\mathsmaller 1$s-Row-Column Deletion} in the $SC1S$ category, and, \textsc{Sc$ \mathsmaller 1$p-$0$-Flipping}, \textsc{Sc$ \mathsmaller 1$p-$1$-Flipping} and \textsc{Sc$\mathsmaller1$p-$01$-Flipping} in the $SC1E$ category. Given a binary matrix $M$, the \textit{\textsc{Sc$\mathsmaller1$s-Row$/$Column$/$Row-Column Deletion}} finds a minimum number of rows/columns/rows as well as columns, whose deletion results in a matrix satisfying the $SC1P$. On the other hand, the \textit{\textsc{Sc$ \mathsmaller 1$p-$0/1/01$-Flipping}} finds a minimum number of $0$-entries/$1$-entries/any entries, to be flipped to satisfy the $SC1P$. We show that the decision versions of the above defined problems are NP-complete. We refer to the parameterized versions of the above problems, parameterized by $d$ as $d$-$SC1S$-$R/C/RC$ and $d$-$SC1P$-$0E/1E/01E$ respectively, with $d$ being the number of rows$/$columns/rows as well as columns that can be deleted, and the number of $0$-entries/$1$-entries/any entries that can be flipped respectively. \\\\ \textbf{Parameterized Complexity}: Fixed-parameter tractability is one of the ways to deal with NP-hard problems. In parameterized complexity, the running time of an algorithm is measured not only in terms of the input size, but also in terms of a parameter. A parameter is an integer associated with an instance of a problem. It is a measure of some property of the input instance. A problem is \textit{fixed-parameter tractable (FPT)} with respect to a parameter $d$, if there exists an algorithm that solves the problem in $f(d).n^{O(1)}$ time, where $f$ is a computable function depending only on $d$, and $n$ is the size of the input instance. The time complexity of such algorithms can be expressed as $O^{*}(f(d))$, by hiding the polynomial terms in $n$. We recommend the interested reader to \cite{downey2013fundamentals} for a more comprehensive overview of the topic.\\\\ \textbf{Problem Definition}: A \textit{matrix} can be considered as a set of rows (columns) together with an order on this set \cite{dom2009recognition}. Here, in this paper, the term matrix always refer to a \textit{binary matrix}. For a given matrix $M$, $m_{ij}$ refers to the entry corresponding to $i^{th}$ row and $j^{th}$ column of $M$. Matrix having at most $x$ ones in each column and at most $y$ ones in each row is denoted as $(x,y)$-matrix. A $(2,*)$-matrix can contain at most two ones per column and there is no bound on the number of ones per row. A $(*,2)$-matrix has no restriction on the number of ones per column and have at most two ones per row. Given an $m \times n$ matrix $M$, let $R(M)=\lbrace r_{1}, r_{2}, \ldots, r_{m} \rbrace$ and $C(M)=\lbrace c_{1}, c_{2}, \ldots, c_{n} \rbrace$ denote the sets of rows and columns of $M$, respectively. Here, $r_{i}$ and $c_{j}$  denote the binary vectors corresponding to row $\vec{r_{i}}$ and column $\vec{c_{j}}$ of $M$, respectively.  For a subset $R^{'} \subseteq R(M)$ of rows, $M[R^{'}]$ and $M \backslash R^{'}$ denote the submatrix induced on $R^{'}$ and $R(M) \backslash R^{'}$ respectively. Similarly, for a subset $C^{'} \subseteq C(M)$ of columns, the submatrix induced on $C^{'}$ and $C(M) \backslash C^{'}$ are denoted by $M[C^{'}]$ and $M \backslash C^{'}$ respectively. Let $A(M)= \lbrace ij \mid m_{ij}=1 \rbrace$  and $B(M)= \lbrace ij \mid m_{ij}=0 \rbrace$ be the set of indices of all $1$-entries and $0$-entries respectively in $M$. We present the formal definitions of the problems $d$-$SC1S$-$R$, $d$-$SC1S$-$C$, $d$-$SC1S$-$RC$, $d$-$SC1P$-$0E$ and $d$-$SC1P$-$01E$ as follows.\\\\
\noindent\fbox{
    \parbox{5.20 in}{
\textbf{\hspace{0.7 in}Simultaneous Consecutive Ones Submatrix ($SC1S$) Problems}\\  
\textbf{\textit{Instance:}} 	$<M, d>$- An $ m \times n$ matrix $M$ and an  integer $d \geq 0$.\\
\textbf{\textit{Parameter:}} $d$.\\
\textbf{\textit{d-SC1S-R:}} Does there exist a set $R^{'} \subseteq R(M)$, with $\vert R^{'} \vert \leq d$ such that $M \backslash R^{'}$ satisfies the $SC1P$?\\
\textbf{\textit{d-SC1S-C:}} Does there exist a set $C^{'} \subseteq C(M)$, with $\vert C^{'} \vert \leq d$ such that $M \backslash C^{'}$ satisfies the $SC1P$?\\
\textbf{\textit{d-SC1S-RC:}} Does there exist sets $ R^{'} \subseteq R(M)$, and $ C^{'} \subseteq C(M)$, with $\vert R^{'} \vert + \vert C^{'} \vert \leq d$ such that  $((M \backslash  R^{'}) \backslash C^{'})$ satisfies the $SC1P$?
}
}
\newline
\noindent\fbox{%
    \parbox{5.20 in}{%
\textbf{\hspace{0.7 in} Simultaneous Consecutive Ones Editing ($SC1E$) Problems }  \\
\textbf{\textit{Instance:}} $<M, d>$- An $ m \times n$ matrix $M$ and an  integer $d \geq 0$.\\
\textbf{\textit{Parameter:}} $d$.\\
\textbf{\textit{d-SC1P-1E \textup{\cite{oswald2009simultaneous}}:}} Does there exist a set $A^{'} \subseteq A(M)$,  with $\vert A^{'} \vert\leq d$ such that the resultant  matrix obtained by flipping the entries of  $A^{'}$ in $M$ satisfies the $SC1P$?
\\
\textbf{\textit{d-SC1P-0E:}} Does there exist a set $B^{'} \subseteq B(M)$, with $\vert B^{'} \vert \leq d$ such that the resultant  matrix obtained by flipping the entries of  $B^{'}$ in $M$ satisfies the $SC1P$?    \\
\textbf{\textit{d-SC1P-01E:}} Does there exist a set $I \subseteq A(M) \cup B(M)$, with $\vert I \vert \leq d$ such that the resultant matrix obtained by flipping the entries of $I$ in $M$ satisfies the $SC1P$?
}
}
\paragraph{} \textbf{Complexity Status:} Oswald and Reinelt \cite{oswald2009simultaneous} posed the decision version of the \textsc{Sc$ \mathsmaller 1$p-$1$-Flipping} problem as \textit{$k$-augmented simultaneous consecutive ones property} and showed that it is NP-complete even for $(*,2)$-matrices. To the best of our knowledge, the parameterized problems posed under $SC1S$ and $SC1E$ category are not explicitly mentioned in the literature. Also, the classical complexity and parameterized complexity of $SC1S$ and $SC1E$ problems are not known prior to this work.\\ 

\noindent \textbf{Our Results:} We investigate the classical complexity and fixed-parameter tractability of $SC1S$ and $SC1E$ problems (defined above). We prove the NP-completeness of the decision versions of $SC1S$ and $SC1E$ problems except for the \textsc{Sc$ \mathsmaller 1$p-$1$-Flipping} problem. Using bounded search tree technique, few reduction rules and related results from the literature \cite{cao2016chordal,kaplan1999tractability}, we present fixed-parameter tractable algorithms for $d$-$SC1S$-$R$, $d$-$SC1S$-$C$, $d$-$SC1S$-$RC$ and $d$-$SC1P$-$0E$ problems on \textit{general matrices} (where there is no restriction on the number of ones in rows and columns) with run-times $O^{*}(8^{d})$, $O^{*}(8^{d})$, $O^{*}(2^{O(dlogd)})$ and $O^{*}(18^{d})$ respectively. 

For $(2,2)$-matrices, we observe that $SC1S$ and $SC1E$ problems are solvable in polynomial-time. We also give improved FPT algorithms for $SC1S$ and $SC1E$ problems on certain restricted matrices. We summarize our FPT results in the following table.
\begin{center}
\begin{tabular}{ |p{3.3cm}|p{2.80cm}|p{2.80cm}|p{2.80cm}|}
 \hline
 \color{black}\textbf{Problem} &\color{black}\textbf{($2,*$)-matrix} &\color{black}\textbf{($*,2$)-matrix}&\color{black}\textbf{($*,*$)-matrix}\\
 \hline 
  $d$-$SC1S$-$R/C$ &$O^{*}(4^{d}/3^{d})$ & $O^{*}(3^{d}/4^{d})$ & $O^{*}(8^{d})$\\
 $d$-$SC1S$-$RC$ &$O^{*}(7^{d})$&$O^{*}(7^{d})$&$O^{*}(2^{O(dlogd)})$\\
 $d$-$SC1P$-$0E$ &irrelevant&irrelevant&$O^{*}(18^{d})$\\
 $d$-$SC1P$-$1E$ &$O^{*}(6^{d})$&$O^{*}(6^{d})$&?\\ 
 $d$-$SC1P$-$01E$ &irrelevant&irrelevant&?\\
  \hline
\end{tabular}
\end{center}
\noindent Here, we observe that while defining $d$-$SC1P$-$0E$ and $d$-$SC1P$-$01E$ problems on $(2,*)/(*,2)$-matrix, flipping of $0$-entries, may change the input matrix to one which is not a $(2,*)/(*,2)$-matrix. We also observe that on $(2,*)$-matrices and $(*,2)$-matrices, $SC1S$ and $SC1E$ problems, except \textsc{Sc$ \mathsmaller 1$p-$0$-Flipping} and \textsc{Sc$ \mathsmaller 1$p-$01$-Flipping}, admit constant factor polynomial-time approximation algorithms. \\\\
\textbf{Motivation : } In Bioinformatics \cite{konig2006discovering}, to discover functionally meaningful patterns from a vast amount of gene expression data, one needs to construct the metabolic network of genes using knowledge about their interaction behavior.  A \textit{metabolic network} is made up of all chemical reactions that involve metabolites, and a \textit{metabolite} is the intermediate end product of metabolism. To obtain functional gene expression patterns from this metabolic network, an adjacency matrix of metabolites is created, and clusters of ones are located in the adjacency matrix. One way to find the clusters of ones is to transform the adjacency matrix into a matrix having the $SC1P$ by flipping $0$'s to $1$'s. This practically motivated problem is posed as an instance of the $d$-$SC1P$-$0E$ problem as follows: \\\\
\noindent\fbox{%
    \parbox{5.20 in}{%
\textbf{\hspace{1.6 in} Finding Clusters of Ones}  \\
\textbf{\textit{Instance:}} $<M,d>$, where $M$ is an adjacency matrix of metabolites and $d \geq 0$.\\
\textbf{\textit{Parameter:}} $d$.\\
\textbf{\textit{Question:}} Does there exist a set of $0$-entries of size at most $d$ in $M$, whose flipping results in a matrix with the $SC1P$?
    }
}\\\\
The fixed-parameter tractability of $d$-$SC1P$-$0E$ problem shows that finding clusters of ones from metabolic networks is also FPT.  Another theoretically motivated problem in the area of Graph theory is \textsc{Biconvex Deletion}. An immediate consequence of the fixed-parameter tractability of $d$-$SC1S$-$RC$ problem is that \textsc{Biconvex Deletion} problem is FPT. In addition, the fixed-parameter tractability of $d$-$SC1P$-$0E$ problem shows that \textsc{Biconvex Completion} problem is also FPT. Several practically relevant problems (scheduling, matching, etc \cite{lipski1981efficient, peng2007treewidth}) are polynomial-time solvable on biconvex graphs. \\\\
\noindent\fbox{%
 \label{bicnvx}   \parbox{5.20 in}{%
\textbf{\hspace{1.3 in} Problems on Biconvex Graphs}  \\
\textbf{\textit{Instance:}} $<G,d>$, where $G$=$(V_{\mathsmaller{1}},V_{\mathsmaller{2}},E)$ is a bipartite graph with $\vert V_{1} \vert =n, \vert V_{\mathsmaller{2}} \vert=m$ and $d \geq 0$.\\
\textbf{\textit{Parameter:}} $d$.\\
{\textit{\textsc{Biconvex Deletion}:}} Does there exist a set $D \subseteq V_{\mathsmaller{1}} \cup V_{\mathsmaller{2}}$, with $\vert D \vert \leq d$ such that $G[(V_{\mathsmaller{1}} \cup V_{\mathsmaller{2}}) \backslash D]$ is a biconvex graph ?\\
\textit{\textsc{Biconvex Edge Deletion:}} Does there exist a set $E' \subseteq E$, with $\vert E' \vert \leq d$ such that $G$=$(V_{\mathsmaller{1}},V_{\mathsmaller{2}}, E \backslash E')$ is a biconvex graph ?\\
\textit{\textsc{Biconvex Completion:}} Does there exist a set $E' \subseteq (V_{\mathsmaller{1}} \times V_{\mathsmaller{2}}) \backslash E$, with $\vert E' \vert \leq d$ such that $G$=$(V_{\mathsmaller{1}},V_{\mathsmaller{2}}, E \cup E')$ is a biconvex graph ?
    }
}\\\\ 

In addition, the FPT algorithm for $d$-$SC1S$-$R$ on $(2,*)$-matrices shows that \textsc{Proper Interval Vertex Deletion} (Section \ref{proper}) problem on \textit{triangle-free} graphs is FPT (using Lemma \ref{lem1}) with a run-time of $O^{*}(4^{d})$, where $d$ denotes the number of allowed vertex deletions. The FPT algorithm for $d$-$SC1P$-$1E$ on $(2,*)$-matrices shows that \textsc{Biconvex Edge Deletion} problem is fixed-parameter tractable on certain  bipartite graphs, in which the degree of all vertices in one partition is at most two.\\\\
\textbf{Techniques Used:} Our results rely on the following forbidden submatrix characterization of the $SC1P$ (see Figure \ref{forbidden}) by Tucker \cite{tucker1972structure}.\begin{theorem}\label{thm2} \textup{(\cite [Theorem 11] {tucker1972structure})} A matrix $M$ has the $SC1P$ if and only if no submatrix of $M$, or of the transpose of $M$, is a member of the configuration (see Section \ref{config})  
 of $M_{I_{k}}(k \geq 1)$, $M_{2_{1}}$, $M_{2_{2}}$, $M_{3_{1}}$, $M_{3_{2}}$  and $M_{3_{3}}$.\end{theorem} \noindent That is, a matrix $M$ has the $SC1P$ if and only if no submatrix of $M$ is a member of the configuration of $M_{I_{k}} (k \geq 1)$, $M_{2_{1}}$,  $M_{2_{2}}$, $M_{3_{1}}$, $M_{3_{2}}$,  $M_{3_{3}}$ or their transposes. We refer to the set of all forbidden submatrices of the $SC1P$ as $F_{SC1P}$.
 
\noindent $F_{SC1P}=\{M_{\mathsmaller{I_{k}}}, M_{\mathsmaller{2_{1}}}, M_{\mathsmaller{2_{2}}}, M_{\mathsmaller{3_{1}}}, M_{\mathsmaller{3_{2}}}, M_{\mathsmaller{3_{3}}}, M^{T}_{\mathsmaller{I_{k}}}, {M^{T}_{\mathsmaller{2_{1}}}}, {M^{T}_{\mathsmaller{2_{2}}}}, {M^{T}_{\mathsmaller{3_{1}}}}, {M^{T}_{\mathsmaller{3_{2}}}}, {M^{T}_{\mathsmaller{3_{3}}}}\}$, where $k \geq 1$.\label{fsc1p}\\  
 \begin{figure}[t]
  \begin{minipage}{.35\linewidth}
    \centering
    \[\left[\begin{array}{ccccccc}
      1 & 1 & 0 & . & . & . & 0\\   
      0 & 1 & 1 & 0 & . & . & 0\\
      0 & 0 & 1 & 1 & 0 & . & 0\\
      . & . & . & . & . & . & .\\
      0 & . & . & . & 0 & 1 & 1\\
      1 & 0 & . & . & . & 0 & 1\\
    \end{array}\right]\]
   \textbf{ $M_{I_{k}}$, $k \geq 1$ ($k+2$ rows and $ k+2$ columns) }
  \end{minipage}%
  \begin{minipage}{.3\linewidth}
    \centering
    \[\left[\begin{array}{cccc}
      1 & 1 & 0 & 0\\    
      0 & 1 & 1 & 0\\
      0 & 1 & 1 & 1\\
      1 & 1 & 0 & 1\\
    \end{array}\right]\]
    \textbf{\hspace{0.0 in}$M_{2_{1}}$}
  \end{minipage}
  \begin{minipage}{.3\linewidth}
    \centering
    \[\left[\begin{array}{ccccc}
     1 & 1 & 0 & 0 & 0\\   
     0 & 1 & 1 & 0 & 0\\
     0 & 0 & 1 & 1 & 0\\
     0 & 1 & 1 & 1 & 1\\
     1 & 1 & 1 & 0 & 1\\       
    \end{array}\right]\]
    \textbf{$M_{2_{2}}$}
   \end{minipage}
   
   \begin{minipage}{.35\linewidth}
    \centering
    \[\left[\begin{array}{cccc}
      1 & 1 & 0 & 0 \\          
      0 & 1 & 1 & 0 \\
      0 & 1 & 0 & 1 \\
    \end{array}\right]\]
    \textbf{$M_{3_{1}}$}
   \end{minipage}
   \begin{minipage}{.3\linewidth}
    \centering
    \[\left[\begin{array}{ccccc}
      1 & 1 & 0 & 0 & 0\\  
      0 & 1 & 1 & 0 & 0\\  
      0 & 0 & 1 & 1 & 0\\  
      0 & 1 & 1 & 0 & 1\\          
    \end{array}\right]\]
  \noindent \textbf{$M_{3_{2}}$}
   \end{minipage}
   \begin{minipage}{.3\linewidth}
    \centering
    \[\left[\begin{array}{cccccc}
      1 & 1 & 0 & 0 & 0 & 0 \\    
      0 & 1 & 1 & 0 & 0 & 0 \\  
      0 & 0 & 1 & 1 & 0 & 0 \\ 
      0 & 0 & 0 & 1 & 1 & 0 \\ 
      0 & 1 & 1 & 1 & 0 & 1 \\ 
    \end{array}\right]\]
    \textbf{\hspace{0.3 in}$M_{3_{3}}$}
   \end{minipage}
  \caption{A subset of the forbidden submatrices for the $SC1P$  \cite{tucker1972structure}.\label{forbidden}}
\end{figure}\\

\noindent For a given matrix $M$, while solving $SC1S$ and $SC1E$ problems, a recursive branching algorithm first destroys all fixed size forbidden submatrices from $F_{SC1P}$. For $d$-$SC1S$-$R/C/RC$  and $d$-$SC1P$-$0E$ problems, the number of branches for this step will be at most $6/6/11$ and $18$ respectively. If the resultant matrix still does not have the $SC1P$, then the only forbidden submatrices that can remain in $M$ are of type $M_{\mathsmaller{I_{k}}}$ and $M^{T}_{\mathsmaller{I_{k}}}$, where $k \geq 1$. 

In $d$-$SC1S$-$R/C$ and $d$-$SC1S$-$RC$ problems, we reduce the resultant matrix at each leaf node of the bounded search tree to an instance of $d$-$COS$-$R$ (Section \ref{dcosr}) and \textsc{Chordal Vertex Deletion} (Section \ref{chrdldef}) problems respectively. Then, we apply algorithms of \textit{$d$-$COS$-$R$} (Theorem \ref{lemma11}) and \textsc{Chordal Vertex Deletion} (Theorem \ref{thmchrdl}) problems to the reduced instances of $d$-$SC1S$-$R/C$ and $d$-$SC1S$-$RC$ problems respectively. Finally, the output of $d$-$SC1S$-$R/C$ and $d$-$SC1S$-$RC$ problems on $M$, relies on the output of $d$-$COS$-$R$ and Chordal-Vertex-Deletion algorithms respectively on the reduced instances. 

For $d$-$SC1P$-$0E$ problem, we prove in Section \ref{etrowcolumn} that, the presence of a large $M_{\mathsmaller{I_{k}}}/ M^{T}_{\mathsmaller{I_{k}}}$ (where $k > d$) is enough to say that we are dealing with a No instance, but this is not the case, for $d$-$SC1S$-$R/C/RC$ and $d$-$SC1P$-$1E/01E$ problems. Using a result on the number of $4$-cycle decompositions of an even $n$-cycle where $n \geq 6$, from \cite{kaplan1999tractability}, we show in Section \ref{etrowcolumn} that, the number of ways to destroy an $M_{\mathsmaller{I_{k}}}/ M^{T}_{\mathsmaller{I_{k}}}$ (where $k \leq d$) in $d$-$SC1P$-$0E$ is equal to the number of ternary trees with $k$-$1$ internal nodes, which is crucial for our FPT algorithm. We prove in Section \ref{fpt0e} that, the number of ternary trees with $k$-$1$ internal nodes can be improved from $8^{k-1}$ to $6.75^{k-1}$, using Stirlings approximation (Lemma \ref{main2}). \\\\
\noindent \textbf{Organization of the paper: }In Section \ref{prel}, we provide necessary preliminaries and observations. Section \ref{esly} presents polynomial-time algorithms for $SC1S$ and $SC1E$ problems on $(2,2)$-matrices. The classical complexity and fixed-parameter tractability of $SC1S$ and $SC1E$ problems are described in Sections  \ref{estrc} and \ref{etrowcolumn} respectively. Last section draws conclusions and gives an insight to further work.
\section{Preliminaries}\label{prel}
In this section, we present definitions and notations related to binary matrix and graphs associated with binary matrix. We recall the definition of a few graph classes that are related to the $SC1P$. For the sake of completeness, we also define some commonly known matrices that are used to represent graphs. We also state a few results that are used in proving the NP-completeness and fixed-parameter tractability of the problems posed in Section \ref{intro}.

\subsection{\textbf{Graphs}} \noindent A \textit{graph} $G$ is defined as a tuple $G=(V, E)$, where $V=\{v_{1},v_{2}, \ldots ,v_{n}\}$ is a finite set of vertices and $E=\{e_{1},e_{2}, \ldots ,e_{m}\}$ is a finite set of edges. Throughout this paper, we consider $\vert V \vert =n$ and $\vert E \vert =m$ respectively. All graphs discussed in this paper shall always be undirected and simple. We refer the reader to  \cite{west2001introduction} for the standard definitions and notations related to graphs. A sequence of distinct vertices $(u_{1},u_{2}, \ldots ,u_{n})$ with $u_{i}$ adjacent to $u_{i+1}$ for each $1 \leq i < n $ is called a $u_{1}$-$u_{n}$ \textit{path}. A \textit{Hamiltonian path} is a path that visits every vertex exactly once. A \textit{cycle} is a graph consisting of a path $(u_{1},u_{2}, \ldots u_{n})$ and the additional edge $\{u_{n},u_{1}\}$. The \textit{length of a path (cycle)} is the number of edges present in it. A cycle (path) on $n$ vertices is denoted as $C_{n}$ ($P_{n}$). Two vertices $u,w$ in $V$ are \textit{connected}, if there exists a path between $u$ and $w$ in $G$. A graph $G=(V,E)$ is a \textit{connected graph}, if there exists a path between every pair of vertices in $V$. A graph $G^{'}=(V^{'}, E^{'})$  is a \textit{subgraph} of $G$, if $V^{'}\subseteq V$ and $E^{'} \subseteq E$. The subgraph of $G$ \textit{induced} by $V^{'}$, denoted as $G[V^{'}]$, is the graph $G^{'}=(V^{'}, E^{'})$ with $V^{'} \subseteq V$  and  $E^{'} = \{\{v, w\} \in E \mid v \in V^{'}$  and $w \in V^{'} \}$.  A graph $G=(V,E)$ is called a \textit{triangle-free}  (\textit{$C_{3}$-free}) graph, if it does not contain $C_{3}$ as an induced subgraph. A \textit{connected component} of $G$ is a maximal connected subgraph of $G$. \textit{Deletion} of a vertex $v \in V$ means, deleting $v$ and all edges incident on $v$. 

A \textit{chord} in a  cycle is an edge that is not part of the cycle but connects  two non-consecutive vertices in the cycle. A \textit{hole} or \textit{chordless cycle} is a cycle of length at least four, where no chords exist. In other words, a chordless cycle $C$ is a cycle ($u_{1},u_{2}, \ldots u_{n})$ with $n \geq 4$, and the additional constraint that there exists no edges of the form $(u_{i},u_{j})$, where $j\neq i\pm 1$ and $2 \leq i,j \leq n-1$. A graph is \textit{chordal} if it contains no hole. That is, in a chordal graph, every cycle of length at least four contains a chord. A chord $(u,v)$ is an \textit{odd chord} in an even-chordless cycle $C$, if the number of edges in the paths connecting $u$ and $v$ is odd. For an even-chordless cycle $C$, a \textit{4-cycle decomposition} \label{4-cycle} is a minimal set $O$, of odd chords in $C$, such that $C \cup O$ does not have induced even chordless cycles of length at least six. We need the following lemma for our algorithms described in Section \ref{esly}.   \begin{lemma}\label{prop6}\textup{(\cite[Theorem 2] {uno2014efficient})} In a graph $G=(V,E)$, a chordless cycle can be detected in $O(n + m)$-time, where $n$ and $m$ are the number of vertices and edges in $G$ respectively.\end{lemma}
 \noindent Given a graph $G=(V,E)$, and a non-negative integer $k$, \textsc{Chordal Vertex Deletion} problem decides whether there exists a set of vertices of size at most $k$ in $V$, whose deletion results in a chordal graph. We used the following theorem for our FPT \label{chrdldef} algorithm described in Section \ref{fptrc}.
\begin{theorem}\textup{(\cite[Theorem 1.1] {cao2016chordal})} \label{thmchrdl}
\textsc{Chordal Vertex Deletion} problem is fixed-parameter tractable with a run-time of $O^{*}(2^{klogk})$, where $k$ is the number of allowed vertex deletions.
\end{theorem}

\noindent Here, we define certain graph classes that are related to the $SC1P$.
\begin{definition} \textup{A graph is an \textit{interval graph} if for every vertex, an interval on the real line can be assigned, such that two vertices share an edge, iff their corresponding intervals intersect. A graph is a \textit{proper interval} graph, if it is an interval graph that has an intersection model, in which no interval properly contains another.}\end{definition} \noindent Given a graph $G=(V,E)$, and a non-negative integer $k$, \textsc{Proper Interval Vertex Deletion} \cite{van2013proper} problem \label{proper} decides whether there exists a set of vertices of size at most $k$ in $V$, whose deletion results in a proper interval graph. 
\begin{definition}\textup{A graph $G = (V, E)$ is \textit{bipartite} if $V$ can be partitioned into two disjoint vertex sets $V_{\mathsmaller{1}}$ and $V_{\mathsmaller{2}}$ such that every edge in $E$ has one endpoint in $V_{\mathsmaller{1}}$ and the other endpoint in $V_{\mathsmaller{2}}$. A bipartite graph is denoted as $G=(V_{\mathsmaller{1}},V_{\mathsmaller{2}},E)$, where $V_{\mathsmaller{1}}$ and $V_{\mathsmaller{2}}$  are the two partitions of $V$.}\end{definition} \begin{definition}\textup{A bipartite graph is \textit{chordal bipartite} if each cycle of length at least six has a chord.}\end{definition}\noindent We observe that a bipartite graph $H$, which is an even chordless cycle of length $2n$, where $n \geq 3$ can be converted to a chordal bipartite graph by adding $n$-$2$ edges. This observation is also mentioned in a different form in \textup{(\cite[Lemma 4.2]{kaplan1999tractability})}. The number of ways to achieve this is given in the following lemma.
\begin{lemma}\textup{(\cite[Lemma 4.3]{kaplan1999tractability})}\label{main1}
Given a bipartite graph $H=(V_{\mathsmaller{1}}, V_{\mathsmaller{2}}, E)$, which is an even chordless cycle of length $2n$ (where $n \geq 3$), the number of ways to make $H$ a chordal bipartite graph  by adding $n$-$2$ edges is equal to the number of ternary trees with $n$-$1$ internal nodes and is no greater than $8^{n-1}$.
\end{lemma}
We used the following lemma to get a tighter upper bound of $6.75^{n-1}$ for the number of ways to make $H$ a chordal bipartite graph.
\begin{lemma} \textup{\cite{stirling}}\label{main2} 
$\lim_{n\to\infty} n! = {\sqrt{2\pi{n}}{(\dfrac{n}{e})}^{n}}$ (This is well known as Stirlings approximation). 
\end{lemma} The following lemma gives the number of ternary trees with $n$ internal nodes.
\begin{lemma} (\cite{graham1989concrete}, p.349)\label{tern}
The number of ternary trees with $n$ internal nodes is equal to $\frac {\binom {3n+1} {n}} {3n+1} = \frac {\binom {3n} {n}} {2n+1}$.
\end{lemma}
\begin{definition}\label{biconvex} \textup{A bipartite graph $G=(V_{\mathsmaller{1}},V_{\mathsmaller{2}},E)$ is \textit{biconvex} if the vertices of both $V_{\mathsmaller{1}}$ and $V_{\mathsmaller{2}}$ can be ordered, such that for every vertex $v$ in $V_{\mathsmaller{1}} \cup V_{\mathsmaller{2}}$, the neighbors of $v$ occur consecutively in the ordering.} \end{definition} \noindent Given a bipartite graph $G=(V_{\mathsmaller{1}},V_{\mathsmaller{2}},E)$, and a non-negative integer $k$, \textsc{Biconvex Deletion} problem decides whether there exists a set of vertices of size at most $k$ in $V_{1} \cup V_{2}$, whose deletion results in a biconvex graph. \textsc{Biconvex Deletion} problem can be shown to be NP-complete, using the results given by Yannakakis \cite{yannakakis1978node,yannakakis1981node}. \begin{definition} \label{chain}\textup{A bipartite graph $G=(V_{\mathsmaller{1}},V_{\mathsmaller{2}},E)$ is called a \textit{chain graph} \cite{natanzon2001complexity} if there exists an ordering $\pi$ of the vertices in $V_{\mathsmaller{1}}$, $\pi: \{1, 2, \ldots, \vert V_{\mathsmaller{1}} \vert\} \rightarrow V_{\mathsmaller{1}}$ such that $N(\pi(1)) \subseteq N(\pi(2)) \subseteq \ldots \subseteq N(\pi(\vert V_{\mathsmaller{1}} \vert))$, where $N(\pi(i))$ denotes the set of neighbours of $\pi(i)$ in $G$.}\end{definition} \noindent Given a bipartite graph $G=(V_{\mathsmaller{1}}, V_{\mathsmaller{2}}, E)$, and a non-negative integer $k$,  \textsc{$k$-Chain Completion} problem decides whether there exists a set of $k$ non-edges in $G$, whose addition transforms $G$ into a chain graph. Yannakakis \cite{yannakakis1981computing} showed that \textsc{$k$-Chain Completion} problem is NP-complete. He also developed finite forbidden induced subgraph characterization for chain graphs. Accordingly, a bipartite graph $G=(V_{\mathsmaller{1}},V_{\mathsmaller{2}},E)$ is a chain graph iff it does not contain $2K_{2}$ as an induced subgraph, where $K_{2}$ is a complete graph on two vertices. Given a bipartite graph, \textsc{$k$-Chain Editing} problem decides whether there exists a set of $k$ edge additions and deletions, which transforms $G$ into a chain graph. Drange et al. \cite{chainEditing} have shown that $k$-\textsc{Chain Editing} problem belongs to the class NP-complete.
\subsection{\textbf{Matrices}} \noindent Given an  $m \times n$ matrix $M$, the $n \times m$ matrix $M^{'}$ with $m_{ji}^{'} = m_{ij}$ is called the \textit{transpose} of $M$ and is denoted by $M^{T}$. Two matrices $M$ and $M^{'}$ are \textit{isomorphic} if $M$ is a permutation of the rows or$/$and columns of $M^{'}$. We say, a matrix $M$ \textit{contains} $M^{'}$, if $M$ contains a submatrix that  is isomorphic to $M^{'}$. The \textit{configuration} \label{config} of an $m \times n$ matrix $M$ is defined to be the set of all  $m \times n$ matrices which can be obtained from $M$ by row or$/$and column permutations. 

\noindent Here, we define some commonly known matrices that are used to represent graphs.
\begin{definition}\label{defn4}\textup{
The \textit{half adjacency matrix} \cite{dom2009recognition} of a bipartite graph $G=(V_{\mathsmaller{1}},V_{\mathsmaller{2}}$ $E)$ with $V_{\mathsmaller{1}} = \{u_{1}, \ldots, u_{n_{1}}\}$ and $V_{\mathsmaller{2}} = \{v_{1},\ldots, v_{n_{2}}\}$ is an $n_{1} \times n_{2}$ matrix $M_{\mathsmaller{G}}$  with $m_{ij} = 1$ iff $\{u_{i} , v_{j} \} \in  E$, where $1 \leq i \leq n_{1}$ and $1 \leq j \leq n_{2}$.}
\end{definition}
\noindent Every matrix $M$ can be viewed as the half adjacency matrix of a bipartite graph. The corresponding bipartite graph of $M$ is referred to as the \textit{representing graph} of $M$, denoted by $G_{\mathsmaller{M}}$. The \textit{representing graph} $G_{\mathsmaller{M}}$ \cite{dom2009recognition} of a matrix $M_{m \times n}$ is obtained as follows:
\begin{definition}\label{defn5}
\textup{For a matrix $M_{m \times n}$, $G_{\mathsmaller{M}}$ contains a vertex corresponding to every row and every column of $M$, and there is an edge between two vertices corresponding to $i^{th}$ row and $j^{th}$ column of $M$ iff the corresponding entry $m_{ij}= 1$, where $1 \leq i \leq m$ and $1 \leq j \leq n$.}
\end{definition}  
Characterizations of biconvex and chain graphs relating their half adjacency matrices are mentioned in Lemma \ref{biconvexlemm} and \ref{chnrs}.
\begin{lemma}\textup{\cite{tucker1972structure}} \label{biconvexlemm}
A bipartite graph is biconvex iff its half adjacency matrix has the $SC1P$.
\end{lemma}
\begin{lemma} \textup{\cite{yannakakis1981computing}}\label{chnrs}
A bipartite graph $G=(V_{\mathsmaller{1}}, V_{\mathsmaller{2}}, E)$ is a chain graph iff  its half adjacency matrix $M_{\mathsmaller{G}}$ does not contain $\begin{pmatrix}1 & 0\\ 0 & 1\end{pmatrix}$ as a submatrix. \end{lemma}
\noindent We remark here that, the half-adjacency matrix of a chain graph satisfies the $SC1P$, however the converse is not true.

\noindent A graph $G$ can also be represented using \textit{edge-vertex incidence matrix}, denoted by $M(G)$, and is defined as follows.
\begin{definition} \label{defn6} \textup{For a graph $G = (V,E)$, the rows and columns of $M(G)$ correspond to edges and vertices of $G$ respectively. The entries $m_{ij}$ of $M(G)$ are defined as follows: $m_{ij}=1$, if edge $e_{i}$ is incident on vertex $v_{j}$, and $m_{ij}=0$ otherwise, where $1 \leq i \leq m$ and $1 \leq j \leq n$.}
\end{definition}
Following Lemma shows that $G$ is a path if $M(G)$ has the $C1P$ for rows.
\begin{lemma}\label{prop4} \textup{(\cite[Theorem 2.2]{dom2009recognition})}
If $G$ is a connected graph and the edge-vertex incidence matrix $M(G)$ of $G$ has the $C1P$ for rows, then $G$ is a path.
\end{lemma}
\noindent A graph $G$ can also be represented using \textit{maximal-clique matrix (\textit{vertex-clique incidence matrix})}, and is defined as follows.
\begin{definition}\label{maxcliq}\textup{Let $V=\{v_{1},v_{2},\ldots,v_{n}\}$  and $C=\{c_{1},c_{2},\ldots,c_{m}\}$  be the set of vertices and the set of maximal cliques, respectively, in $G$. The maximal-clique matrix of $G$ is an $n \times m$ matrix $M$, whose rows and columns represent the vertices and maximal cliques, respectively, in $G$, and an entry $m_{ij}=1$ if $v_{i}$ belongs to $c_{j}$, and $m_{ij}=0$ otherwise, where $1 \leq i \leq n$ and $1 \leq j \leq m$}.
\end{definition}
A characterization of proper interval graph relating its maximal-clique matrix is mentioned in Lemma \ref{lem1}.
 \begin{lemma}\textup{\cite{dom2009recognition}}\label{lem1}
 A graph is a proper interval graph iff its maximal-clique matrix has the $SC1P$.
 \end{lemma}
\noindent Next, we state few results that are used in proving the correctness of our FPT algorithms described in Sections \ref{fptr/c}-\ref{fpt0e}.

\noindent For ease of reference, we refer to the \label{defn1} fixed-size forbidden matrices in the forbidden submatrix characterization of $SC1P$ (Theorem \ref{thm2}) as $X$. i.e 

$X=\{M_{\mathsmaller{2_{1}}}, M_{\mathsmaller{2_{2}}}, M_{\mathsmaller{3_{1}}}, M_{\mathsmaller{3_{2}}}, M_{\mathsmaller{3_{3}}}, {M^{ T}_{\mathsmaller{2_{1}}}}, {M^{T}_{\mathsmaller{2_{2}}}}, {M^{T}_{\mathsmaller{3_{1}}}}, {M^{T}_{\mathsmaller{3_{2}}}}, {M^{T}_{\mathsmaller{3_{3}}}}\}$. 

\noindent Lemma \ref{prop1} and Lemma \ref{prop2} state the run-time to find a forbidden matrix of $X$ and $M_{\mathsmaller I_{k}}/M^{T}_{\mathsmaller{I_{k}}}$ respectively in $M$.
\begin{lemma}\label{prop1}
\textup{} \textit{Let $M$ be a matrix of size $m \times n$. Then, a minimum size submatrix in $M$ that is isomorphic to one of the forbidden matrices of $X$ can be found in $O(m^{6}n)$-time}. \end{lemma}
The above Lemma is obtained from (\cite[Proposition 3.2]{dom2009recognition}), by considering the maximum possible size of the forbidden matrix in $X$ as $6 \times 5$ (shown in Figure \ref{forbidden}). By considering the maximum number of ones in each row of $M$ as $n$ in (\cite[Proposition 3.4]{dom2009recognition}), leads to the following Lemma.\begin{lemma}\label{prop2}Let $M$ be a matrix of size $m \times n$. Then, a minimum size submatrix of type $M_{I_{k}}$ or $M^{T}_{\mathsmaller{I_{k}}}$  ($k \geq 1$) in $M$ can be found in $O(n^{3}m^{3})$-time.\end{lemma}
Following result shows that the representing graph $G_{\mathsmaller{M}_{I_{k}}}/G_{\mathsmaller{M}^{T}_{\mathsmaller{I_{k}}}}$ (Definition \ref{defn5}) of $M_{\mathsmaller I_{k}}/M^{T}_{\mathsmaller{I_{k}}}$ (where $k \geq 1$) is a chordless cycle.
\begin{lemma} \label{prop3} \textup{(\cite[Observation 3.1]{dom2009recognition})} The representing graph of $M_{\mathsmaller I_{k}}/ M^{T}_{\mathsmaller{I_{k}}}$, i.e., 
\textup{($G_{\mathsmaller{M}_{I_{k}}}$/$G_{\mathsmaller{M}^{T}_{\mathsmaller{I_{k}}}}$)}, is a chordless cycle of length $2k+4$.
\end{lemma}
It is clear from Lemma \ref{prop3}, that the representing graph of both $M_{I_{k}}$ and its transpose are same, which simplifies the task of searching for $M_{\mathsmaller I_{k}}/M^{T}_{\mathsmaller{I_{k}}}$. 

Few of our results are based on the forbidden submatrix characterization of the $C1P$ for rows and is given below.
\begin{theorem}\label{forbc1p}\textup{(\cite[Theorem 9]{tucker1972structure})}
A binary matrix $M$ has the $C1P$ for rows if and only if no submatrix of $M$ is a member of the configuration of $M_{I_{k}}$, $M_{II_{k}}$, $M_{III_{k}}$, $M_{IV}$ and $M_{V}$, where $k \geq 1$.
\end{theorem}
Given a binary matrix $M$ and a non-negative integer $d$, $d$-$COS$-$R$ (resp. $d$-$COS$-$C$) \label{dcosr}problem decides whether there exists a set of rows (resp. columns), of size at most $d$ in $M$, whose deletion results in a matrix with the $C1P$ for rows.  

\noindent We used the following lemma to obtain an FPT algorithm for $d$-$SC1S$-$R/C$  problems.
\begin{lemma} \textup{(\cite[Theorem 7]{narayanaswamy2015obtaining})} \label{lemma11}
$d$-$COS$-$R$ problem is fixed-parameter tractable with a run-time of $O^{*}(10^{d})$, where $d$ denotes the number of allowed row deletions.
\end{lemma}
 Using the recent improved FPT algorithm for \textsc{Interval Deletion} problem \cite{cao2016linear}, it turns out that $d$-$COS$-$R$ problem has an improved run-time of $O^{*}(8^{d})$.
\section{Our Results}\label{result}
Even though the number of forbidden submatrices to establish the $SC1P$ is less than the number of forbidden submatrices for the $C1P$, the problems posed in this paper, to obtain the $SC1P$ also turn out to be NP-complete. Firstly, we  present polynomial-time algorithms for $SC1S$ and $SC1E$ problems on $(2,2)$-matrices. For a given matrix $M$, while solving $SC1S$ problems, we delete an entire row/column of every forbidden submatrix present in $M$; hence destroying any forbidden submatrix from $F_{SC1P}$ (defined in Section \ref{fsc1p}) in $M$ does not introduce new forbidden submatrices from $F_{SC1P}$ in $M$, which were not originally present in $M$. The same observation, however, is not applicable for $SC1E$ problems. The reason is that flipping an entry ($0/1$) may introduce new forbidden submatrices from $F_{SC1P}$ which were not originally  present in $M$. This motivated us to consider the two categories of problems for establishing the $SC1P$ in a given matrix separately. The classical complexity as well as the parameterized complexity of $SC1S$ and $SC1E$ problems are described in detail in Sections \ref{estrc} and \ref{etrowcolumn} respectively.
\subsection{\textup{\textbf{Easily solvable instances of $SC1S$ and $SC1E$ problems}}} \label{esly}
The problems $SC1S$ and $SC1E$ defined in Section \ref{intro} are solvable in polynomial-time on $(2,2)$-matrices. A $(2, 2)$-matrix can contain only forbidden matrices $M_{\mathsmaller I_{k}}$ and  $M^{T}_{\mathsmaller{I_{k}}}$ (where $k \geq 1$) of unbounded size, because all other forbidden matrices of $F_{SC1P}$ contain either a row or column with more than two ones. Since a matrix can be viewed as the half adjacency matrix of a bipartite graph, the $d$-$SC1S$-$R$, $d$-$SC1S$-$C$, $d$-$SC1S$-$RC$, $d$-$SC1P$-$0E$, $d$-$SC1P$-$1E$, and $d$-$SC1P$-$01E$ problems can be formulated as graph modification problems (Here, modification means deletion of vertex/edge or addition of edge).

 Given a $(2, 2)$-matrix $M$, consider the representing graph $G_{\mathsmaller{M}}$ (Definition \ref{defn5}), of $M$. Since each column and row of $M$ contains at most two ones, the degree of each vertex in $G_{\mathsmaller{M}}$ is at most two. So the connected components of $G_{\mathsmaller{M}}$ are disjoint chordless cycles or paths. It follows from Lemma \ref{prop3} that, to destroy $M_{\mathsmaller I_{k}}$ and $M^{T}_{\mathsmaller{I_{k}}}$, it is sufficient to destroy chordless cycles of length greater than four in $G_{\mathsmaller{M}}$.
\begin{theorem}\label{easily}
On $(2,2)$-matrices, $d$-$SC1S$-$R$ is polynomial-time solvable.
\end{theorem}
\begin{proof}
  For each chordless cycle $C$ of length greater than four in $G_{\mathsmaller{M}}$, consider the submatrix $M'$ induced by the vertices of $C$. To destroy $C$, delete a vertex $v$ in $C$, that corresponds to a row $r$ in $M'$. Decrement the parameter $d$ by one and delete $r$ from $M$. The input is an Yes-instance, if the total number of rows removed from $M$ is at most $d$, otherwise it is a No-instance.
  
The representing graph $G_{\mathsmaller{M}}$ of $M$ can be constructed in polynomial time. Since the degree of each vertex in $G_{\mathsmaller{M}}$ is at most two, every pair of chordless cycles in $G_{\mathsmaller{M}}$ will be disjoint. We also know that $G_{\mathsmaller{M}}$ contains only finite number of vertices. The above two facts imply that $G_{\mathsmaller{M}}$ contains only finite number of cycles. Using Lemma \ref{prop6}, each chordless cycle can be detected in $O(m + n)$-time. Therefore for $(2,2)$-matrices, $d$-$SC1S$-$R$ can be solved in $O(d(m + n))$-time.
\end{proof}
 Algorithms for solving $d$-$SC1S$-$C$, $d$-$SC1S$-$RC$, $d$-$SC1P$-$0E$, $d$-$SC1P$-$1E$ and $d$-$SC1P$-$01E$ problems on $(2,2)$-matrices are similar to the algorithm for solving $d$-$SC1S$-$R$ (Theorem \ref{easily}), except that they differ only in the way the chordless cycles are destroyed. Therefore the run-time of all these problems on $(2,2)$-matrices is $O(d(m + n))$. Let $C$ be a chordless cycle of length greater than four in $G_{\mathsmaller{M}}$. In the following corollaries, we describe how the chordless cycles are destroyed in each of the problems.
\begin{corollary} 
For $(2,2)$-matrices, $d$-$SC1S$-$C$ problem is polynomial-time solvable.
\end{corollary}
\begin{proof}
 In $d$-$SC1S$-$C$ problem, deletion of a column in $M$ corresponds to a vertex deletion in the representing graph $G_{\mathsmaller M}$. For each chordless cycle $C$ in $G_{\mathsmaller{M}}$, consider the submatrix $M'$ induced by the vertices of $C$. To destroy $C$, delete a vertex $v$ in $C$, that corresponds to a column in $M'$.   
 \end{proof}
 \begin{corollary}
For $(2,2)$-matrices, $d$-$SC1S$-$RC$ problem is polynomial-time solvable.
\end{corollary}
\begin{proof}
In $d$-$SC1S$-$RC$ problem, deletion of a row as well as column in $M$ corresponds to a vertex deletion in the representing graph $G_{\mathsmaller M}$. For each chordless cycle $C$ in $G_{\mathsmaller{M}}$, consider the submatrix $M'$ induced by the vertices of $C$. To destroy $C$, delete a vertex $v$ in $C$, that corresponds to a row or column in $M'$. 
\end{proof}
\begin{corollary}
For $(2,2)$-matrices, $d$-$SC1P$-$0E$ problem is polynomial-time solvable.
\end{corollary}
\begin{proof}
In $d$-$SC1P$-$0E$ problem, flipping a $0$-entry in $M$ corresponds to an edge addition in the representing graph $G_{\mathsmaller{M}}$. For each chordless cycle $C$ of length, say $k$, in $G_{\mathsmaller{M}}$, consider the submatrix $M'$ induced by the vertices of $C$. From Lemma \ref{main1}, to make $C$ a chordal bipartite graph, we have to add $\frac {k} {2}$-$2$ edges. Hence, check whether the parameter $d \geq \frac {k} {2}$-$2$ or not. If so, decrement $d$ by $\frac {k} {2}$-$2$ and flip the $0$-entries corresponding to the newly added edges in $M$. The input is an Yes-instance if the total number of $0$-entries flipped in $M$ (edges added in $G_{\mathsmaller{M}}$) is at most $d$, otherwise it is a No-instance. 
\end{proof}
\begin{corollary}\label{cor1}
For $(2,2)$-matrices, $d$-$SC1P$-$1E$ problem is polynomial-time solvable.
\end{corollary}
\begin{proof}
In $d$-$SC1P$-$1E$ problem, flipping a $1$-entry in $M$ corresponds to an edge deletion in the representing graph $G_{\mathsmaller{M}}$. To destroy $C$, delete an edge, say $e$ in $C$. Decrement the parameter $d$ by one and flip the corresponding $1$-entry in $M$. The input is an Yes-instance if the total number of $1$-entries flipped in $M$ (edges deleted in $G_{\mathsmaller{M}}$) is at most $d$, otherwise it is a No-instance. 
\end{proof}
\begin{corollary}
For $(2,2)$-matrices, $d$-$SC1P$-$01E$ problem is polynomial-time solvable.
\end{corollary}
\begin{proof}
 In $d$-$SC1P$-$01E$ problem, the allowed operations are edge additions and edge deletions. In a chordless cycle $C$ of length $2k+4$, the number of edges to be added to destroy $C$ is $k$ where $k \geq 1$, but deletion of any edge in $C$ destroys $C$. Hence, we always delete an edge from each of the chordless cycles in $G_{\mathsmaller{M}}$ for destroying it. This proof is same as the proof of Corollary \ref{cor1}. 
\end{proof}
\subsection{\textup{\textbf{Establishing $SC1P$ by Deletion of Rows/Columns}}} \label{estrc}
This section considers the classical complexity and fixed-parameter tractability of $SC1S$ problems by row, column and row as well as column deletion. We refer to the decision versions of the optimization problems \textsc{Sc$\mathsmaller1$s-Row Deletion}, \textsc{Sc$\mathsmaller 1$s-Column Deletion} and \textsc{Sc$\mathsmaller 1$s-Row-Column Deletion} defined in Section \ref{opt} as $k$-$SC1S$-$R$, $k$-$SC1S$-$C$, and $k$-$SC1S$-$RC$ respectively, where $k$ denotes the number of allowed deletions. First, we show that these problems are NP-complete. Then, we give FPT algorithms for these problems on general matrices. For each of these  problems, we also give improved FPT algorithms on certain restricted matrices.
\subsubsection{\textup{\textbf{NP-Completeness}}}\label{hard1}
The following theorem proves the NP-completeness of $k$-$SC1S$-$R$ problem using \textsc{Hamiltonian-Path} as a candidate problem. 
\begin{theorem}\label{thm14}
Given an $m \times n$ matrix $M$, deciding if there exists a set $R^{'} \subseteq R(M)$, of rows such that $\vert R^{'} \vert \leq k$ and $M \backslash R^{'}$ have the $SC1P$ is NP-complete.
\end{theorem}
\begin{proof}
We first show that $k$-$SC1S$-$R$ $\in$ NP. Given a matrix $M$ and an integer $k$, the certificate chosen is a set of rows $R^{'} \subseteq R(M)$. The verification algorithm affirms that $\vert R^{'} \vert \leq k $, and then it checks  whether deletion of these $k$ rows from $M$ yields a matrix with the $SC1P$. This certificate can be verified in polynomial-time.\\
\indent We prove that $k$-$SC1S$-$R$ problem is NP-hard by showing that \textsc{Hamiltonian-Path} $\leq_{\mathsmaller {P}}$  $k$-$SC1S$-$R$. Let $G=(V,E)$ be a graph with $\vert V \vert = n$ and $\vert E \vert =m$, and $M(G)_{m \times n}$ be the edge-vertex incidence matrix (see Definition \ref{defn6}) obtained from $G$. Without loss of generality, assume that $G$ is connected and let $k$ be $m$-$n$+$1$. We show that $G$ has a Hamiltonian path if and only if there exists a set of rows of size $k$ in $M(G)$ whose deletion results in a matrix $M^{'}(G)$, that satisfy the $SC1P$.

Assume that $G$ contains a Hamiltonian path. In $M(G)$, delete the rows that correspond to edges which are not part of the Hamiltonian path in $G$. Since Hamiltonian path contains $n$-$1$ edges, the number of rows remaining in $M(G)$ will be $n$-$1$ which is equal to $m$-$k$ and hence the number of rows deleted will be $k$. Now, order the columns and rows of $M(G)$ with respect to the sequence of  vertices and edges, respectively in the Hamiltonian path. Clearly, the resulting matrix has the $SC1P$.

To prove the other direction, let $M^{'}(G)$ be the matrix obtained by deleting $k$ rows from $M(G)$ and assume that $M^{'}(G)$ has the $SC1P$. Now, the number of rows in $M^{'}(G)$ is  $m$-$k$, which is equal to $n$-$1$. Let $G^{'}$ be the subgraph obtained from $M^{'}(G)$, by considering $M^{'}(G)$ as an edge-vertex incidence matrix of $G^{'}$. Since $M^{'}(G)$ has the $SC1P$; it has the $C1P$ for rows. Also, note that $M^{'}(G)$ has $n$-$1$ rows. We claim that the subgraph $G^{'}$ is connected. Otherwise one of the connected components of $G^{'}$ must contain a cycle which contradicts the fact that $M^{'}(G)$ has the $C1P$ for rows. This implies that $G^{'}$ is a path (see Lemma \ref{prop4}) of length $n$-$1$, which clearly indicates that $G$ has a Hamiltonian path. The column permutation needed to convert $M^{'}(G)$ into a matrix that has the $C1P$ for rows gives the relative order of vertices of $G$'s Hamiltonian path.
 This proves the NP-completeness of $k$-$SC1S$-$R$. 
\end{proof}
\begin{corollary}
The problem $k$-$SC1S$-$C$ is NP-complete.
\end{corollary}
\begin{proof}
The NP-completeness of $k$-$SC1S$-$C$ can be proved similar to Theorem \ref{thm14} (NP-completeness of $k$-$SC1S$-$R$) by considering $M$ as the vertex-edge incidence matrix and $k$ as the number of columns to be deleted. 
\end{proof}
\noindent Since the edge-vertex incidence matrix (resp. vertex-edge incidence matrix) is a $(*,2)$-matrix (resp. $(2,*)$-matrix), in fact, the following stronger result holds: \begin{corollary}
$k$-$SC1S$-$R$ (resp. $k$-$SC1S$-$C$) problem is NP-complete even for $(*,2)$-matrices (resp. $(2,*)$-matrices).
\end{corollary}
\noindent To prove the NP-completeness of the $k$-$SC1S$-$RC$ problem, we use the \textsc{Biconvex Deletion} problem (Definition \ref{biconvex}) as a candidate problem. The following theorem proves the NP-completeness of $k$-$SC1S$-$RC$. 
\begin{theorem}
The $k$-$SC1S$-$RC$ problem is NP-complete.
\end{theorem}
\begin{proof}
It is easy to show that $k$-$SC1S$-$RC \in NP$. We prove that $k$-$SC1S$-$RC$ problem is NP-hard by showing that \textsc{Biconvex Deletion} problem $\leq_{\mathsmaller {P}}$  $k$-$SC1S$-$RC$. Let $G=(V_{\mathsmaller{1}}, V_{\mathsmaller{2}}, E)$ be a bipartite graph and $M$ be a half adjacency matrix (see Definition \ref{defn5}) of $G$. Using Lemma \ref{biconvexlemm}, it can be shown that $G$ has a set of vertices, $V'_{1} \subseteq V_{\mathsmaller{1}}$ and $V'_{2} \subseteq V_{\mathsmaller{2}}$, with $\vert V'_{1} \vert + \vert V'_{2} \vert \leq k$, whose deletion results in a biconvex graph if and only if there exists a set of rows $R' \subseteq R(M)$ and columns $C' \subseteq C(M)$, with $\vert R' \vert + \vert C' \vert \leq k$ in $M$ whose deletion results in a matrix $M'$, that satisfy the $SC1P$. Therefore $k$-$SC1S$-$RC$ is NP-complete.  
\end{proof}
\subsubsection{\textup{\textbf{An FPT algorithm for $d$-$SC1S$-$R/d$-$SC1S$-$C$ problem}} \label{fptr/c}} 
\noindent Here, we present an FPT algorithm \textit{$d$-$SC1S$-$Row$-Deletion} (Algorithm \ref{alg5}), for $d$-$SC1S$-$R$ problem on general matrices. Given a binary matrix $M$ and a non-negative integer $d$, Algorithm \hyperref[alg5]{1} first destroys the fixed size forbidden submatrices from $X$ in $M$, using a simple search tree based branching algorithm. If $M$ contains a forbidden matrix from $X$ (see Section \ref{defn1}), then the algorithm recursively branches into at most six subcases, since the largest forbidden matrix of $X$ has six rows. In each subcase, delete one of the rows of the forbidden submatrix of $X$ found in $M$ and decrement the parameter $d$ by one. This process is continued in each subcase until its $d$ value becomes zero or until it does not contain any matrix from $X$ as its submatrix. If any of the leaf instances satisfy the $SC1P$, then algorithm returns Yes, indicating that input is an Yes instance. Otherwise, for each  valid leaf instance (leaf instances with $d_{i} >0$), say $\langle M_{i}, d_{i} \rangle$ (where $1 \leq i \leq 6^{d}$) of the above depth bounded search tree, if $M_{i}$ still does not have the $SC1P$, then destroy $M_{I_{k}}$ and $M^{T}_{\mathsmaller{I_{k}}}$ (where $k \geq 1$) in $M_{i}$, using the algorithm for $d$-$COS$-$R$ (see Lemma \ref{lemma11}) on $M_{i}$. The following claim holds true for any leaf instance $\langle M_{i}, d_{i} \rangle$, where $1 \leq i \leq 6^{d}$.
  \begin{algorithm}[h]
\caption{Algorithm  \textit{$d$-$SC1S$-$Row$-Deletion}$(M,d)$ \label{alg5}}
\begin{algorithmic}[1] 
 \Require An instance $\langle M_{m \times n},d \rangle$, where $M$ is a binary matrix and $d \geq 0$. \vspace{-0.1 in} \[ \hspace{-0.48 in}\textbf{Output} = \begin{dcases*} Yes,  & if there exists a set  $R^{'} \subseteq R(M)$, with $\vert R^{'} \vert \leq d$, such that $M \backslash R'$ \\  
        & has the $SC1P$.  
   \\  
   No,    & otherwise
\end{dcases*}
\] \vspace{-0.2 in}
 \State \textbf{if} {$M$ has the $SC1P$ and $d \geq 0$} \textbf{then} return Yes. \vspace{-0.05 in}
 \State \textbf{if}  {$d < 0$} \textbf{then} return No. \vspace{-0.05 in} \newline \vspace{-0.05 in} 
 \noindent \textit{\textbf{Branching Step:}}
\State \textbf{if} {$M$ contains a forbidden submatrix $M'$ from $X$}, \vspace{-0.05 in}
    \newline \vspace{-0.05 in}
  \indent Branch into at most $6$  instances $I_{i}=\langle M_{i}, d_{i}\rangle$ where $i \in \{1,2,\ldots,6\}$ \vspace{-0.03 in} \newline \vspace{-0.03in}
  \indent $M_{i} = M \backslash r_{i}$, where $r_{i} \in R(M')$ \vspace{-0.03 in} \newline  \vspace{-0.03 in}
 \indent Update $d_{i} = d -1$  \hspace{1.0 in}// Decrement parameter by 1. \vspace{-0.03 in} \newline \vspace{-0.03 in} 
\noindent For some $i \in \{1,2,\ldots,6\}$, if \textit{$d$-$SC1S$-$Row$-Deletion}$(M_{i},d_{i})$ return Yes, then return Yes, else if all instances return No, then return No.  \vspace{-0.05 in}
\State \textbf{else}\vspace{-0.05 in}
\State \indent \textbf{if} {$M$ contains either $M_{I_{k}}$ or $M^{T}_{\mathsmaller{I_{k}}}$},\vspace{-0.05 in}
\State \indent \indent $D=$ $d$-$COS$-$R(M,d)$ (Using Lemma \ref{lemma11})\vspace{-0.05 in}
\State \indent \indent \textbf{if} {$\vert D \vert > 0$ }\vspace{-0.05 in}
\State \indent \indent \indent return Yes\vspace{-0.05 in}
\State \indent \indent \textbf{else}\vspace{-0.05 in}
\State  \indent \indent \indent return No\vspace{-0.05 in}
\State \indent \indent \textbf{end if}\vspace{-0.05 in}
\State \indent \textbf{end if}\vspace{-0.05 in}
\State \textbf{end if} \vspace{-0.05 in}
\end{algorithmic}
\end{algorithm}
\begin{claim}
Let $M$ be a matrix that does not contain any fixed size forbidden matrices from $X$, then $d$-$COS$-$R(M,d)$ would destroy only forbidden matrices of the form $M_{I_{k}}$ and $M^{T}_{\mathsmaller{I_{k}}}$ in $M$, where $k \geq 1$.
\end{claim}
\begin{proof}
Let $F_{\mathsmaller {C1PR}}$, $F_{\mathsmaller {C1PC}}$, and $F_{\mathsmaller {SC1P}}$ represent the set of forbidden submatrices of $C1P$ for rows, $C1P$ for columns and $SC1P$ respectively.

Let $F_{\mathsmaller {C1PR}}$= $X_{1} \cup \{M_{I_{k}}\}$, where $X_{1}=\{M_{II_{k}}, M_{III_{k}}, M_{IV}, M_{V}\}$ (see Theorem \ref{forbc1p})

Then, $F_{\mathsmaller {C1PC}}$= $X^{T}_{1} \cup \{M^{T}_{\mathsmaller{I_{k}}}\}$, where $X^{T}_{1}=\{M^{T}_{II_{k}}, M^{T}_{\tiny {III_{k}}}, M^{T}_{IV}, M^{T}_{V}\}$

Now, $F_{\mathsmaller {SC1P}}= X_{1} \cup X^{T}_{1} \cup \{M_{I_{k}}, M^{T}_{\mathsmaller{I_{k}}}\}$

\noindent From Lemma \ref{prop3}, it is clear that searching for both $M_{I_{k}}$ and its transpose is equivalent to searching for $M_{I_{k}}$ alone. 

This implies, $F_{\mathsmaller {SC1P}}= X_{1} \cup X^{T}_{1} \cup \{M_{I_{k}}\}$, where $k \geq 1$

Now, one of the matrices from $X$ occurs as a submatrix of every matrix in $X_{1} \cup X^{T}_{1}$. Since $M$ being a matrix not containing any matrices from $X$, $M$ will not have any matrix from $X_{1} \cup X^{T}_{1}$ as a submatrix. Hence, employing $d$-$COS$-$R$ on $M$ would destroy only forbidden submatrices of the form $M_{I_{k}}$ in $M$.
\end{proof}
If any of the valid leaf instances $\langle M_{i}, d_{i} \rangle$ (where $1 \leq i \leq 6^{d}$) return Yes after employing $d$-$COS$-$R$ algorithm, then Algorithm \ref{alg5} returns Yes indicating that $M$ is an Yes instance, otherwise it returns No. 

\begin{theorem}\label{thmr}
$d$-$SC1S$-$R$ is fixed-parameter tractable on general matrices with a run-time of $O^{*}(8^{d})$. 
\end{theorem}
\begin{proof}
Algorithm \hyperref[alg5]{1} employs a search tree, in which each node in the tree has at most six subproblems. Let us assume that out of the $d$ row-deletions that are allowed, $d_1$ are used for destroying the finite size forbidden matrices, and $d_2$ are used to destroy the remaining non-finite forbidden matrices. Therefore, the tree has at most $6^{d_{1}}$ leaves. A submatrix $M^{'}$ of $M$, that is isomorphic to one of the forbidden matrices in $X$ can be found in $O(m^{6}n)$-time (using Lemma \ref{prop1}). Therefore, the time taken to destroy the finite size forbidden matrices is $O^{*}$($6^{d_1}$). For each leaf instance, destroying all $M_{\mathsmaller I_{k}}$ and $M^{T}_{\mathsmaller{I_{k}}}$ (where $k \geq 1$) using $d$-$COS$-$R$ subroutine (Lemma \ref{lemma11}) takes $O^{*}(8^{d})$-time.    Therefore, the time taken to destroy the non-finite size forbidden matrices is $O^{*}$($8^{d_2}$). So, the total run-time of the algorithm would be $O^{*}$($6^{d_1}$.$8^{d_2}$)=$O^{*}$($8^d$).
\end{proof}
\noindent Since $d$-$SC1S$-$C$ on $M$ is equivalent to $d$-$SC1S$-$R$ on $M^{T}$, we obtain the following corollary.
\begin{corollary}
$d$-$SC1S$-$C$ is fixed-parameter tractable on general matrices with a run-time of $O^{*}(8^{d})$.
\end{corollary} 
\subsubsection{\textup{\textbf{An FPT algorithm for $d$-$SC1S$-$RC$ problem}}} \label{fptrc}
Here, we present an $FPT$ algorithm \textit{$d$-$SC1S$-$RC$-Deletion} (Algorithm \ref{alg6}), for the problem \begin{algorithm}[h]
\caption{Algorithm  \textit{$d$-$SC1S$-$RC$-Deletion}$(M,d)$\label{alg6}}
\begin{algorithmic}[1] 
 \Require An instance $\langle M_{m \times n},d \rangle$, where $M$ is a binary matrix and $d \geq 0$.\vspace{-0.1 in} \[ \hspace{-0.26 in}\textbf{Output} = \begin{dcases*} Yes,  & if there exists a set $R^{'} \subseteq R(M)$ and $C^{'} \subseteq C(M)$, with $\vert R^{'} \vert + \vert C^{'} \vert \leq d$, \\ & such that  $((M \backslash  R^{'}) \backslash C^{'})$ has the $SC1P$  \\  
   No,     & otherwise
\end{dcases*}
\] \vspace{-0.15 in}
 \State \textbf{if} {$M$ has the $SC1P$ and $d \geq 0$} \textbf{then} return Yes. \vspace{-0.08 in}
 \State \textbf{if}  {$d < 0$} \textbf{then} return No. \vspace{-0.05 in} \newline \vspace{-0.08 in}
  \noindent \textit{\textbf{Branching Step:}}
 \If {$M$ contains a forbidden submatrix $M'$ from $X$},  
\vspace{-0.05 in} \newline \vspace{-0.05 in}
  \indent Branch into at most $11$  instances $I_{i}=\langle M_{i}, d_{i}\rangle$ where $i \in \{1,2,\ldots,11\}$ \vspace{-0.05 in} \newline \vspace{-0.05 in}
  \indent $M_{i} = M \backslash m_{i}$, where $m_{i} \in R(M')$ or $m_{i} \in C(M')$ \vspace{-0.05 in} \newline \vspace{-0.05 in}  
 \indent Update $d_{i} = d -1$  \hspace{1.0 in}// Decrement parameter by 1. \vspace{-0.05 in} \newline \vspace{-0.05 in}
\noindent For some $i \in \{1,2,\ldots,11\}$, if \textit{$d$-$SC1S$-$RC$-Deletion}$(M_{i},d_{i})$ return Yes, then return Yes, else if all instances return No, then return No.  
\Else \vspace{-0.08 in}
\If {$Stage$-$2$($G_{M},d)$ returns Yes},  \hspace{0.2 in}// Presented in Algorithm \ref{alg7} \vspace{-0.08 in}
\State return Yes \vspace{-0.08 in}
\Else \vspace{-0.08 in}
\State return No. \vspace{-0.08 in}
\EndIf \vspace{-0.08 in}
\EndIf
\end{algorithmic}
\end{algorithm} $d$-$SC1S$-$RC$ on general matrices. Algorithm \ref{alg6} consists of two stages. Given a binary matrix $M$ and a non-negative integer $d$, stage $1$ of Algorithm \ref{alg6} destroys all forbidden submatrices from $X$ in $M$ using a simple search tree  algorithm. If $M$ contains a forbidden matrix from $X$, then Algorithm \ref{alg6} branches into at most $11$ subcases, since the number of rows and columns in the largest forbidden matrix of $X$ is $11$. In each subcase, delete one of the rows or columns of the forbidden submatrix found in $M$ and decrement the parameter $d$ by one. This process is continued in each subcase until its $d$ value becomes zero or until it does not contain any matrix from $X$ as its submatrix. If any of the leaf instances satisfy the $SC1P$, then this algorithm returns Yes. Otherwise, to each valid leaf instance $ \langle M_{i},d_{i} \rangle$ (leaf instance with $d_{i}>0$), where $1 \leq i \leq 11^{d}$, we apply stage $2$ of Algorithm \ref{alg6} to destroy $M_{\mathsmaller I_{k}}$ and $M^{T}_{\mathsmaller{I_{k}}}$, where $k \geq 1$. Stage $2$ of Algorithm \ref{alg6} considers the representing graph $G_{\mathsmaller M_{i}}$, of each valid leaf instance $M_{i}$, where $1 \leq i \leq 11^{d}$. The following observation holds true for the representing graph $G_{\mathsmaller M_{i}}$ of each valid leaf instance $M_{i}$.
\begin{Observation}\label{obsr}
Let $M$ be a matrix that does not contain any forbidden matrix in $X$. Then, the representing graph $G_{\mathsmaller M}$, of $M$ contains none of the graphs $G_{ \tiny M_{2_{1}}}, G_{ \tiny M_{2_{1}}^{\tiny T}}$,$G_{ \tiny M_{2_{2}}},G_{ \tiny M_{2_{2}}^{\tiny T}}$, $G_{ \tiny M_{3_{1}}},G_{ \tiny M_{3_{1}}^{\tiny T}}$, $G_{ \tiny M_{3_{2}}},G_{ \tiny M_{3_{2}}^{\tiny T}}$, $G_{ \tiny M_{3_{3}}},G_{ \tiny M_{3_{3}}^{\tiny T}}$ shown in Figure \ref{forbidden subgraph} as its induced subgraph.\end{Observation}
\begin{figure}[h]
 \begin{center} \includegraphics[width=3.0 in]{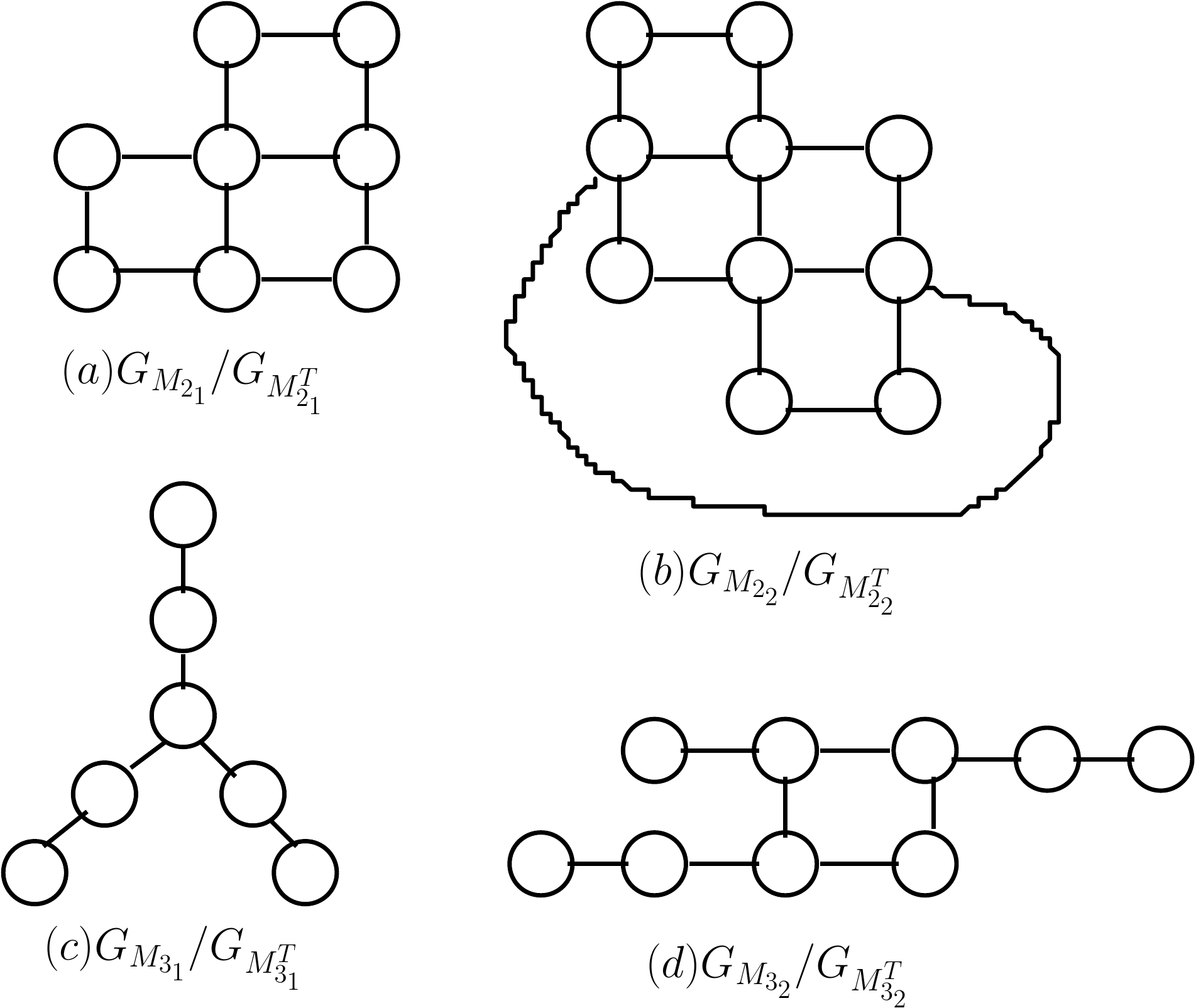}
 \includegraphics[width=1.0 in]{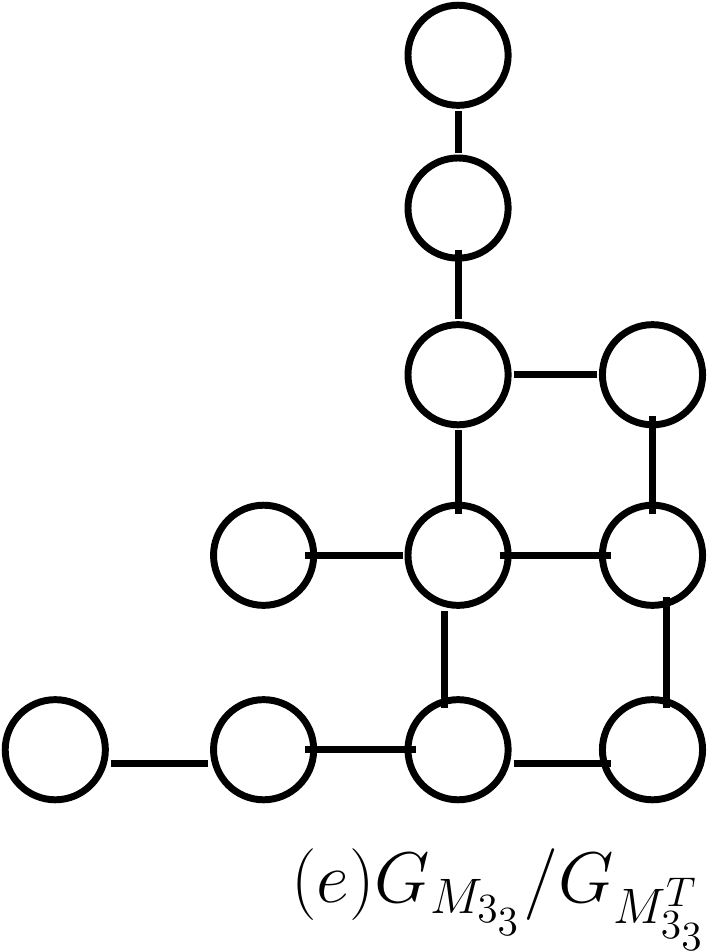} 
 \end{center}
  \caption{Representing graph of (a) $M_{2_{1}}$ ($M^{\tiny T}_{2_{1}}$) (b) $M_{2_{2}}$ ($M^{\tiny T}_{2_{2}}$) (c) $M_{3_{1}}$ ($M^{\tiny T}_{3_{1}}$) (d) $M_{3_{2}}$ ($M^{\tiny T}_{3_{2}}$)  \hspace{.95 in}(e) $M_{3_{3}}$ ($M^{\tiny T}_{3_{3}}$).}
  \label{forbidden subgraph}
\end{figure}
 \newpage It is easy to see that deleting a row or column in $M_{i}$ is equivalent to deleting a vertex in $G_{\mathsmaller M_{i}}$ and, destroying $M_{I_{k}}$ and $M^{T}_{\mathsmaller{I_{k}}}$, where $k \geq 1$ in $M_{i}$ is equivalent to destroying even chordless cycles of length greater than or equal to six in $G_{\mathsmaller M_{i}}$ (Using Lemma \ref{prop3}). This instance is same as that of \textsc{Chordal Vertex Deletion} instance (Section \ref{thmchrdl}) except the fact that $4$-cycles need to be preserved and the remaining chordless cycles are of length greater than or equal to six. Thus in stage $2$ (Algorithm \ref{alg7}), after preserving $4$-cycles in $G_{\mathsmaller M_{i}}$, we use \textit{chordal vertex deletion algorithm} (Theorem \ref{thmchrdl}) to destroy all chordless cycles of length greater than or equal to six. We apply the following reduction rules to $G_{\mathsmaller M_{i}}$ before performing chordal vertex deletion algoirthm on $G_{\mathsmaller M_{i}}$. \\
  \begin{algorithm}[h]
\caption{Algorithm  \textit{STAGE-2}$(G_{M},d)$\label{alg7}}
\begin{algorithmic}[1] 
 \Require An instance $\langle G_{M},d \rangle$, where $G_{M}=(V_{\mathsmaller 1}, V_{\mathsmaller 2}, E)$ is a bipartite graph with $\vert V_{\mathsmaller 1} \vert=m$, $\vert V_{\mathsmaller 2} \vert=n$, such that $G_{M}$ does not contain any of the graphs shown in Figure \ref{forbidden subgraph} as its induced subgraph, and $d \geq 0$. \vspace{-0.1 in} \[  \hspace{-0.23 in}\textbf{Output} = \begin{dcases*} Yes,  & if there exists a set $V^{'} \subseteq V_{\mathsmaller 1} \cup V_{\mathsmaller 2}$, with $\vert V^{'} \vert \leq d$, such that $G_{M} \backslash V'$  \\    
        &  is chordal bipartite.\\
   No,     & otherwise
\end{dcases*}
\]
\State \textbf{if} {$G_{M}$ is chordal bipartite} \textbf{then} return Yes. \vspace{-0.05 in}
 \State \textbf{if} {$d < 0$} \textbf{then} return No. \vspace{-0.05 in}
 \State \textbf{if} {$G_{M}$ contains a chordless cycle $C'$ of length six, eight or ten} \textbf{then}, // Rule \hyperref[presv]{1} \newline \vspace{-0.05 in}
  \indent Branch into at most $10$ instances $I_{i}=\langle G_{M_{i}}, d_{i}\rangle$ where $i \in \{1,2,\ldots,10\}$ \vspace{-0.03 in} \newline \vspace{-0.05 in} 
  \indent  $G_{M_{i}} = G_{M} \backslash v_{i}$, where $v_{i}$ is a vertex in $C'$ \newline   \vspace{-0.03 in}
 \indent Update $d_{i} = d-1$  \hspace{1.0 in} // Decrement parameter by $1$ \vspace{-0.05 in}  \newline \vspace{-0.03 in} 
\noindent For some $i \in \{1,2,\ldots,10\}$, if \textit{STAGE-2}$(G_{M_{i}},d_{i})$ returns Yes, then return Yes, else if all instances return No, then return No. \vspace{-0.08 in}
\State \textbf{else}\vspace{-0.08 in}
\State \indent \textbf{if} {there exists a $4$-cycle $C'$ in $G_{M}$} \textbf{then},  \vspace{-0.08 in}
\State \indent \indent $G'_{M} \leftarrow $  graph obtained from $G_{M}$ after reducing $C'$ using Rule \hyperref[presv]{2} \vspace{-0.08 in}
\State \indent \indent \textbf{if }\textit{STAGE-2}($G'_{M}$,$d$) returns Yes, \textbf{then} return Yes, otherwise return No\vspace{-0.08 in}
\State \indent   \textbf{end if}\vspace{-0.08 in}
\State \indent   Apply Rule \hyperref[presv] {3} to $G_{M}$.  // $Degree_{\leq 1}$ rule\vspace{-0.08 in}
\State \indent  \textbf{if} \hyperref[thmchrdl]{$CHORDAL\_VERTEX\_DELETION$($G_{M}$,$d$)} returns Yes, \textbf{then} \vspace{-0.05 in}
 \State \indent  \indent  return Yes\vspace{-0.1 in}
\State \indent  \textbf{ else}\vspace{-0.1 in}
 \State  \indent  \indent  return No\vspace{-0.1 in}
\State \indent  \textbf{end if}\vspace{-0.08 in}
\State \textbf{end if}\vspace{-0.05 in}
\end{algorithmic}
\end{algorithm}
\noindent \noindent \noindent \textbf{Reduction Rules \newline} \label{presv} 
In order to avoid the destruction of $4$-cycles in $G_{\mathsmaller M_{i}}$ by chordal vertex deletion algorithm, we apply the following reduction rules.

\noindent \textbf{Rule $1$: (Killing shorter chordless cycles):}\label{redtnrul1} If graph $G_{\mathsmaller M_{i}}$ contains a chordless cycle of length six, eight or ten, then branch in to at most ten subproblems, deleting in each branch one of the vertices of the chordless cycle found.

\noindent Recursively apply Rule $1$ to $G_{\mathsmaller M_{i}}$, until all chordless cycles of length six, eight and ten are destroyed from it.
 
\noindent \textbf{Rule $2$: ($4$-cycle preserving rule):}\label{redtnrul1}
If graph $G_{M_{i}}$ contains a $4$-cycle, say $(x_{1},y_{1,}x_{2},y_{2})$ as an induced subgraph, modify $G_{M_{i}}$ as follows: Introduce two new vertices and label them as $x_{1}x_{2}$ and $y_{1}y_{2}$. Make all edges incident on ($x_{1}$ or $x_{2}$) /($y_{1}$ or $y_{2}$) to incident on $x_{1}x_{2}$/$y_{1}y_{2}$ and add an edge between $x_{1}x_{2}$ and $y_{1}y_{2}$. Delete the vertices $x_{1}$, $x_{2}$, $y_{1}$ and $y_{2}$ from $G_{M_{i}}$. This is explained in Figure \ref{fig8}.\\
\noindent Each time after applying Rule $2$, call Rule $1$. The main purpose of calling Rule $1$ after preserving every $4$-cycle is to avoid the longer chordless cycle $C$ (chordless cycle having length greater than or equal to twelve) getting totally disappeared from $G_{\mathsmaller M_{i}}$, when $C$ intersects with many $4$-cycles. Recursively apply Rule $2$ to $G_{M_{i}}$, until $G_{M_{i}}$ contains no $4$-cycles.

\noindent \textbf{Rule $3$: ($Degree_{\leq 1}$ rule):} Delete all vertices having degree less than or equal to one in $G_{M_{i}}$.

\noindent Rule $3$ is safe, since vertices having degree less than or equal to one do not contribute to chordless cycles of length greater than four.
\begin{figure}[h] 
 \begin{center}
 \includegraphics[width=5.0 in]{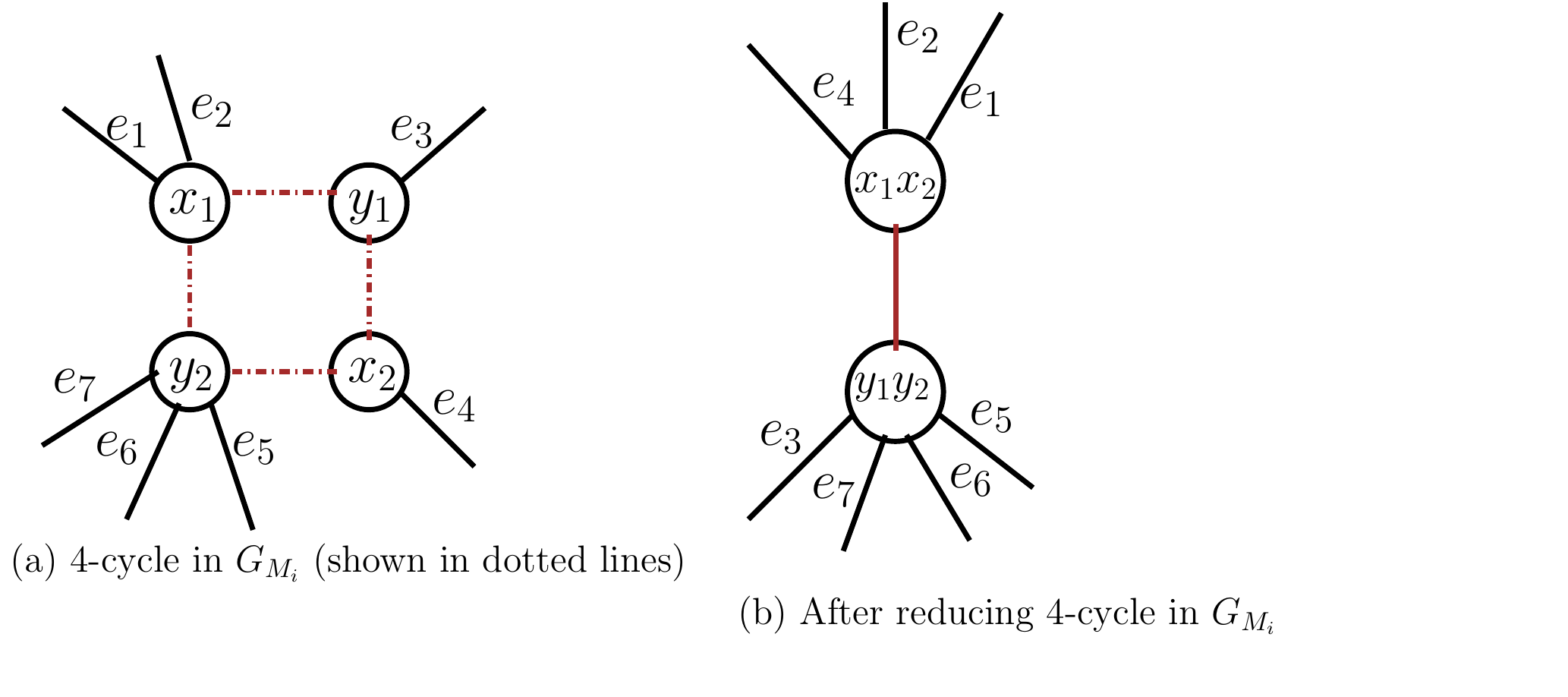} 
 \caption{Illustration of preserving a $4$-cycle in $G_{M_{i}}$}
 \label{fig8}
  \end{center}
   \end{figure}\\
Next, we prove that the process of preserving $4$-cycles in $G_{\mathsmaller M_{i}}$ do not introduce already destroyed forbidden matrices from $X$ in $M_{i}$.
\begin{figure}[t] 
 \begin{center}
 \includegraphics[width=5.0 in]{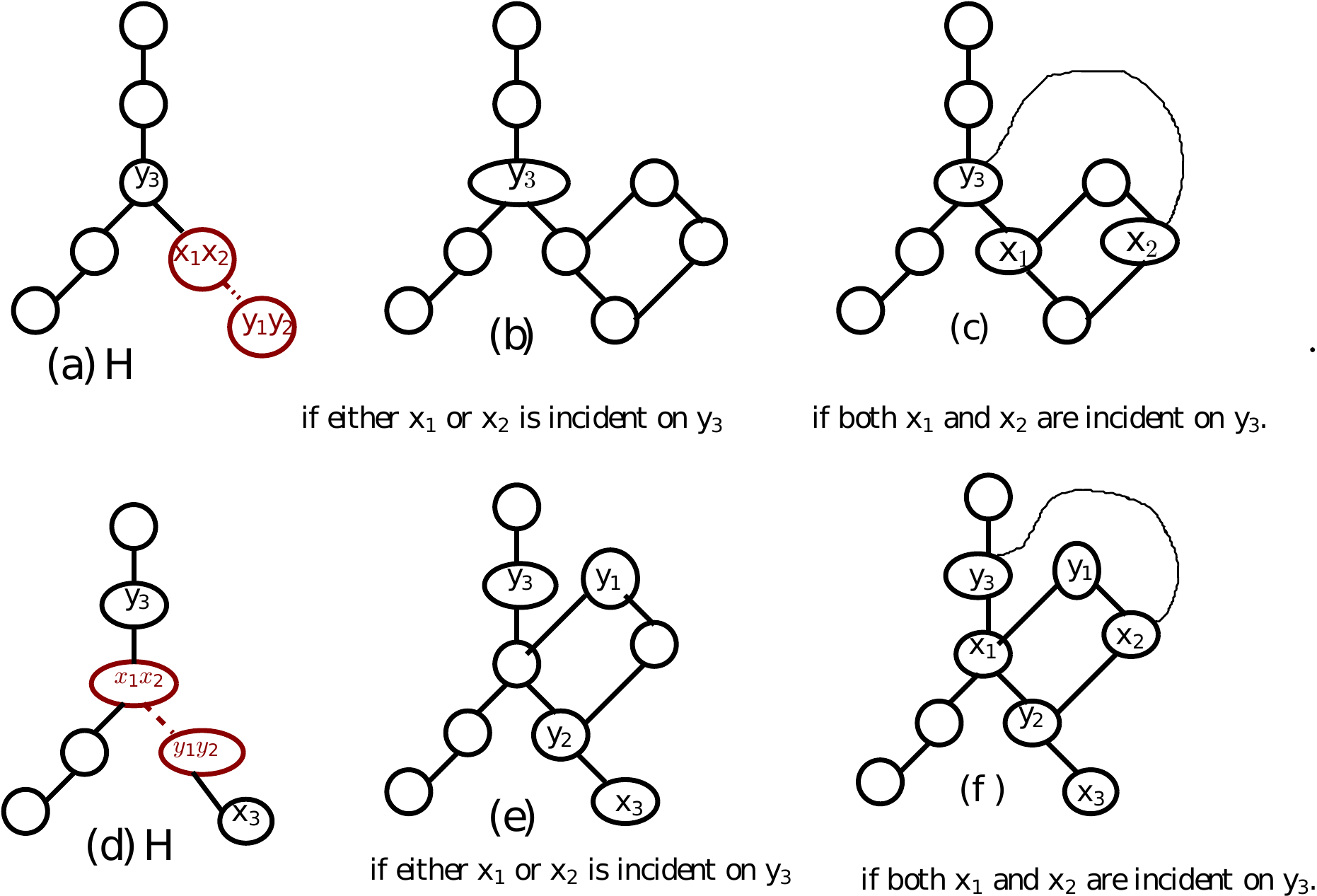} 
 \caption{Illustration of different cases when $G'_{\mathsmaller M_{i}}$ has an induced subgraph isomorphic to $G_{ \tiny M_{3_{1}}}/G_{ \tiny M_{3_{1}}^{\tiny T}}$}
 \label{fig4}
  \end{center}
   \end{figure}
\begin{claim} \label{claim1}
Let $G_{M_{i}}$ be the representing graph of a valid leaf instance $ \langle M_{i},d_{i} \rangle$, obtained after Stage $1$ of Algorithm \ref{alg6} and, let $G'_{M_{i}}$ be the graph obtained from $G_{M_{i}}$, after applying Rules \hyperref[presv]{1, {2}} and \hyperref[presv]{3}. Then, $G'_{M_{i}}$ contains none of the graphs $G_{ \tiny M_{2_{1}}}, G_{ \tiny M_{2_{1}}^{\tiny T}}$,$G_{ \tiny M_{2_{2}}},G_{ \tiny M_{2_{2}}^{\tiny T}}$, $G_{ \tiny M_{3_{1}}},G_{ \tiny M_{3_{1}}^{\tiny T}}$, $G_{ \tiny M_{3_{2}}},G_{ \tiny M_{3_{2}}^{\tiny T}}$, $G_{ \tiny M_{3_{3}}},G_{ \tiny M_{3_{3}}^{\tiny T}}$ as shown in Figure \ref{forbidden subgraph} as an induced subgraph.
\end{claim}
\begin{proof}
The graphs $G_{ \tiny M_{2_{1}}}, G_{ \tiny M_{2_{1}}^{\tiny T}}$,$G_{ \tiny M_{2_{2}}},G_{ \tiny M_{2_{2}}^{\tiny T}}$, $G_{ \tiny M_{3_{2}}},G_{ \tiny M_{3_{2}}^{\tiny T}}$, $G_{ \tiny M_{3_{3}}},G_{ \tiny M_{3_{3}}^{\tiny T}}$ as shown in Figure \ref{forbidden subgraph}(a), (b), (d) and (e) contain chordless cycles of length $4$. Since $4$-cycle preserving rule reduces all chordless cycles of length exactly four from $G_{M_{i}}$, the resultant graph $G'_{M_{i}}$ will not contain any of the graphs $G_{ \tiny M_{2_{1}}}, G_{ \tiny M_{2_{1}}^{\tiny T}}$,$G_{ \tiny M_{2_{2}}},G_{ \tiny M_{2_{2}}^{\tiny T}}$, $G_{ \tiny M_{3_{2}}},G_{ \tiny M_{3_{2}}^{\tiny T}}$, $G_{ \tiny M_{3_{3}}},G_{ \tiny M_{3_{3}}^{\tiny T}}$ as an induced subgraph. Next, we prove that $G'_{M_{i}}$ will not contain any of the graphs $G_{ \tiny M_{3_{1}}}$, $G_{ \tiny M_{3_{1}}^{\tiny T}}$ shown in Figure \ref{forbidden subgraph}(c) as its induced subgraph. For a contradiction, assume that $G'_{M_{i}}$ contains an induced subgraph $H'$, isomorphic to the graph $G_{ \tiny M_{3_{1}}}$ or $G_{ \tiny M_{3_{1}}^{\tiny T}}$. Then, at least one edge, say ($x_{\mathsmaller 1}x_{\mathsmaller 2},y_{\mathsmaller 1}y_{\mathsmaller 2}$) in $H'$ is obtained by reducing a $4$-cycle in $G_{M_{i}}$. Figure \ref{fig4} shows two such cases. The same observation also holds for other edges in $H'$. In each case, it turns out that the original graph $G_{M_{i}}$ contains the graph $G_{ \tiny M_{3_{1}}}$ or $G_{ \tiny M_{3_{1}}^{\tiny T}}$ as its induced subgraph, which is a contradiction (From Observation \ref{obsr}). 
\end{proof}
\begin{figure}[h]
 \begin{center} \includegraphics[width=5.5 in]{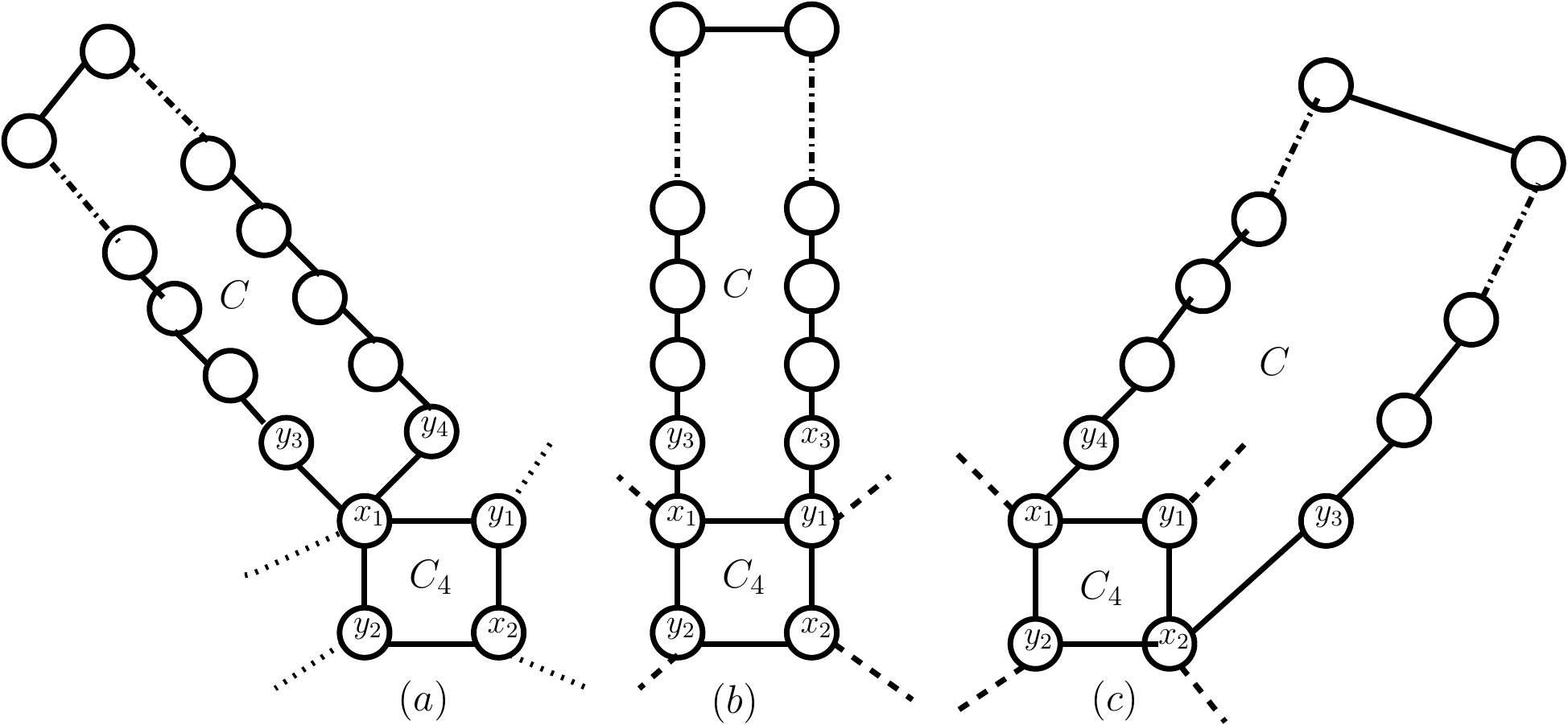}
 \end{center}
  \caption{(a) $C$ contains exactly one vertex from $C_{4}$ (b) $C$ contains exactly two vertices from $C_{4}$ (c) $C$ contains exactly three vertices from $C_{4}$ }
  \label{case 1}
\end{figure}
\noindent Next, we prove that preserving $4$-cycles in $G_{M_{i}}$ using Rule $2$ preserves all existing chordless cycles  of length greater than or equal to $12$ and do not introduce new chordless cycles in $G_{M_{i}}$.
  \begin{figure}[h]
 \begin{center} \includegraphics[width=4.5 in]{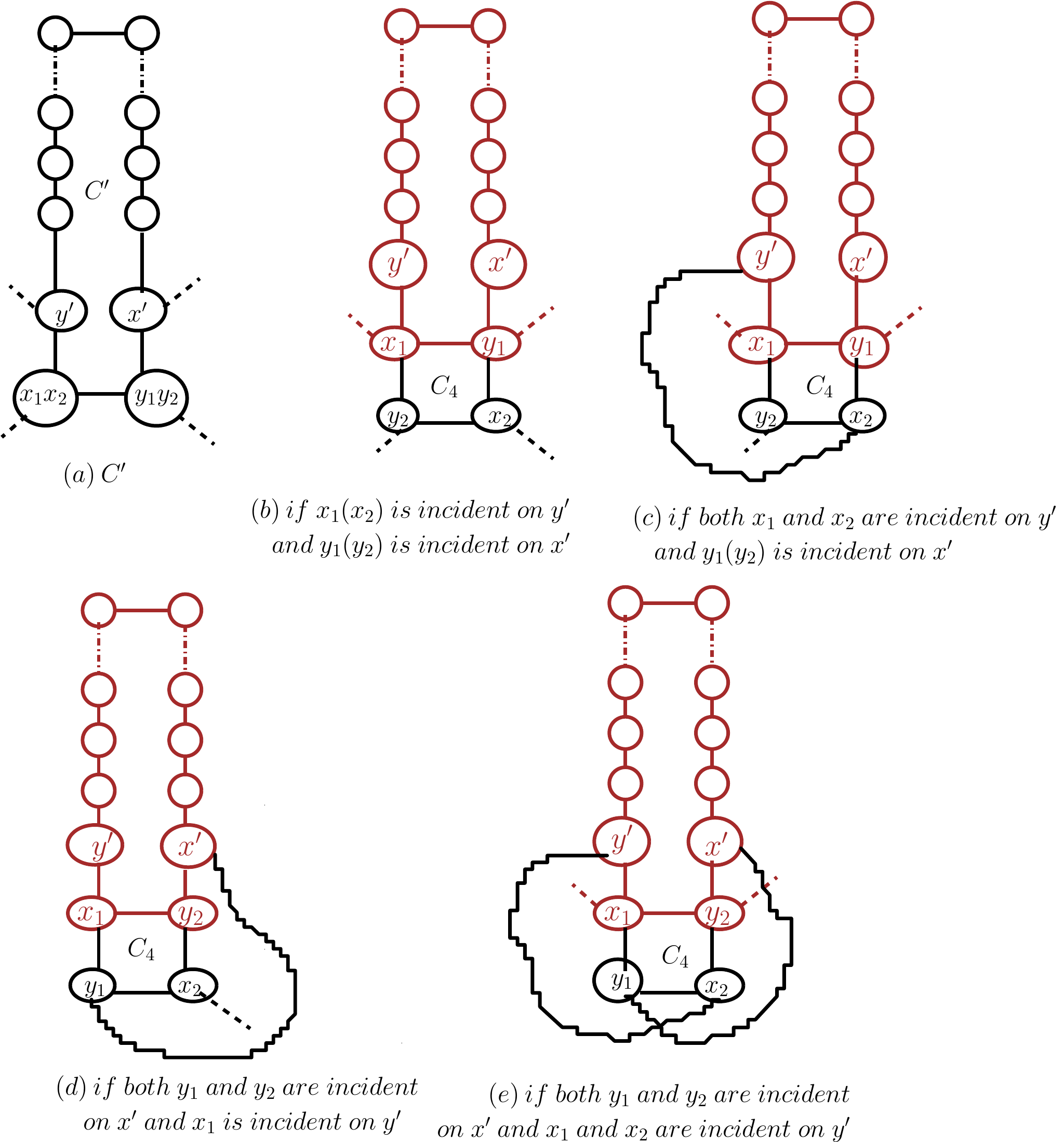}
 \end{center}
  \caption{Illustration of different cases when $x_{1}$, $y_{1}$, $x_{2}$ and $y_{2}$ of $C_{4}$ are incident on $x'$ and $y'$ in $G_{M_{i}}$}
  \label{case2}
\end{figure}
\begin{claim} \label{claim2} 
Let $G_{M_{i}}$ be the representing graph of a valid leaf instance $ \langle M_{i},d_{i} \rangle$, obtained after Stage $1$ of Algorithm \ref{alg6} and, let $C=(x_{1},y_{1},\ldots, x_{n},y_{n})$ be a chordless cycle of length $2n$, where $n \geq 6$ in $G_{M_{i}}$. Let $C_{4}= (x_{1},y_{1},x_{2},y_{2})$ be a chordless cycle of length exactly four in $G_{M_{i}}$. Then, reducing $C_{4}$ using Rule $2$ in $G_{M_{i}}$ preserves $C$ and do not introduce new chordless cycles. It only reduces the length of $C$ by $0, 2, \ldots, 2n-10$.
\end{claim} 
\begin{proof}
Firstly, we show that preserving $4$-cycles in $G_{M_{i}}$, preserves an existing chordless cycle $C$ of length $2n$, where $n \geq 6$.   

\noindent Case $(a)$ : $C$ includes exactly one of the vertices of $C_{4}$, say $x_{1}$ as shown in Figure \ref{case 1}(a). After reducing $C_{4}$, $y_{3}$ and $y_{4}$ will be incident on the newly created vertex $x_{1}x_{2}$. In this case, the length of $C$ is reduced by zero. \newline 
Case $(b)$ : $C$ includes exactly two vertices from $C_{4}$, say $x_{1}$ and $y_{1}$ as shown in Figure \ref{case 1}(b). After reducing $C_{4}$, $y_{3}$ and $x_{3}$ will be incident on the newly created vertices $x_{1}x_{2}$ and $y_{1}y_{2}$ respectively. In this case, the length of $C$ is reduced by zero. \newline 
Case $(c)$ : $C$ includes exactly three vertices from $C_{4}$, say $x_{1}$, $y_{1}$ and $x_{2}$ as shown in Figure \ref{case 1}(c). After reducing $C_{4}$, $y_{3}$ and $y_{4}$ will be incident on the newly created vertex $x_{1}x_{2}$. In this case, the length of $C$ is reduced by two. \newline 
Next, we show that reducing $C_{4}$ using Rule $2$ do not create a new chordless cycle of length greater than or equal to $12$ in $G_{M_{i}}$. For a contradiction, assume that $C'$ be a chordless cycle of length greater than or equal to $12$ in $G'_{M_{i}}$, that is formed as a result of reducing $C_{4}$ using Rule $2$. Then, at least one edge in $C'$ is obtained by reducing $C_{4}$. Let $(x_{1}x_{2},y_{1}y_{2})$ be that edge and without loss of generality assume that $x_{1}x_{2}$ and $y_{1}y_{2}$ are incident on $y'$ and $x'$ respectively in $C'$ as shown in Figure \ref{case2}(a). Then, there should be an induced path of length greater than or equal to $9$ from  $y'$ to $x'$. Figure \ref{case2}(b), (c), (d), and (e) shows the different cases, when $x_{1}$, $y_{1}$, $x_{2}$ and $y_{2}$ are incident on $x'$ and $y'$ in $C'$. In each case, it is easy to see that the induced path from $y'$ to $x'$ along with an edge in $C_{4}$ forms a chordless cycle of length greater than or equal to $12$ in the original graph $G_{M_{i}}$. This implies that $C'$ is not a newly created cycle in $G'_{M_{i}}$.
\end{proof}
\begin{lemma}
Rule $2$ is safe.
\end{lemma}
\begin{proof}
Proof follows from Claim \ref{claim1} and Claim \ref{claim2}.
\end{proof}
\noindent Next, we show that solving $d$-$SC1S$-$RC$ on $M_{i}$ is equivalent to solving chordal vertex deletion problem on $G'_{M_{i}}$.
\begin{lemma}Let $\langle M_{i},d_{i} \rangle$, where $1 \leq i \leq 11^{d}$ be a valid leaf instance obtained after Stage $1$ of Algorithm \ref{alg6} and, let $G_{M_{i}}$ be the representing graph of $M_{i}$. Let $G'_{M_{i}}$ be the graph obtained from $G_{M_{i}}$, after applying Rules \hyperref[presv]{1, {2}} and \hyperref[presv]{3}. Then, solving $d$-$SC1S$-$RC$ on $M_{i}$ is equivalent to solving \textsc{Chordal Vertex Deletion} problem on $G'_{M_{i}}$, and $M_{i}$ has a $d_{i}$ size solution for $d$-$SC1S$-$RC$ if and only if $G'_{M_{i}}$ has a $d_{i}$ size solution for \textsc{Chordal Vertex Deletion} problem.
\end{lemma}
\begin{proof}
The only forbidden matrices that can survive in $M_{i}$ are $M_{I_{k}}$ and $M^{T}_{\mathsmaller{I_{k}}}$ (where $k \geq 1$), which corresponds to even chordless cycles of length greater than or equal to six in $G_{M_{i}}$. Since Rule $2$ reduces each chordless cycle of length exactly four in $G_{M_{i}}$ to an edge in $G'_{M_{i}}$, $G'_{M_{i}}$ will not have any four length chordless cycles. Also, $G'_{M_{i}}$ do not contain odd chordless cycles. Therefore, solving \textsc{Chordal Vertex Deletion} problem on $G'_{M_{i}}$ is equivalent to destroying all $M_{I_{k}}$ and $M^{T}_{\mathsmaller{I_{k}}}$ (where $k \geq 4$) in $M_{i}$, and $M_{i}$ has a $d_{i}$ size solution for $d$-$SC1S$-$RC$ if and only if $G'_{M_{i}}$ has a $d_{i}$ size solution for chordal vertex deletion algorithm.
\end{proof}
\begin{theorem} \label{rest1}
$d$-$SC1S$-$RC$ is fixed-parameter tractable on general matrices with a run-time of $O^{*}(2^{dlogd})$.
\end{theorem}
\begin{proof}
Stage $1$ of Algorithm \ref{alg6} employs a search tree, where each node has at most $11$ subproblems. Therefore, the tree has at most $11^{d}$ leaves after Stage $1$. A submatrix $M^{'}$ of $M$, that is isomorphic to one of the forbidden matrices in $X$ can be found in $O(max(m^{6}n,n^{3}m^{3}))$-time (using Lemma \ref{prop1} and Lemma \ref{prop2}). The initial branching step takes at most $O(11^{d}.max(m^{6}n,n^{3}m^{3}))$-time. Chordal Vertex Deletion algorithm called in each of the leaf instances runs in $O^{*}(2^{dlogd})$-time (Theorem \ref{thmchrdl}). Therefore, the total time complexity of Algorithm \ref{alg6} is $O^{*}(2^{dlogd})$.
\end{proof}
\noindent The following corollary on \textsc{Biconvex Deletion} problem is a direct consequence of Theorem \ref{rest1}.
\begin{corollary}
\textsc{Biconvex Deletion} problem is fixed-parameter tractable on bipartite graphs with a run-time of $O^{*}(2^{dlogd})$, where $d$ denotes the number of allowed vertex deletions.
\end{corollary}
\subsubsection{\textup{\textbf{Improved FPT algorithms for $SC1S$ problems on restricted matrices}}}\label{restrctd}
\begin{figure}[t]
\begin{minipage}{.25\linewidth} 
  \[
\begin{blockarray}{ccccc}
   & & y_{1} & y_{2} & y_{3} \\
\begin{block}{cc (ccc)}
   x_{1} & & 1 & 0 & 0 \\
   x_{2} & & 1 & 1 & 1  \\
   x_{3} & & 0 & 1 & 0 \\
   x_{4} & & 0 & 0 & 1  \\
   \end{block}
\end{blockarray}\]
  \end{minipage}
 \begin{minipage}{.2\linewidth}
\includegraphics[width=1.0 in]{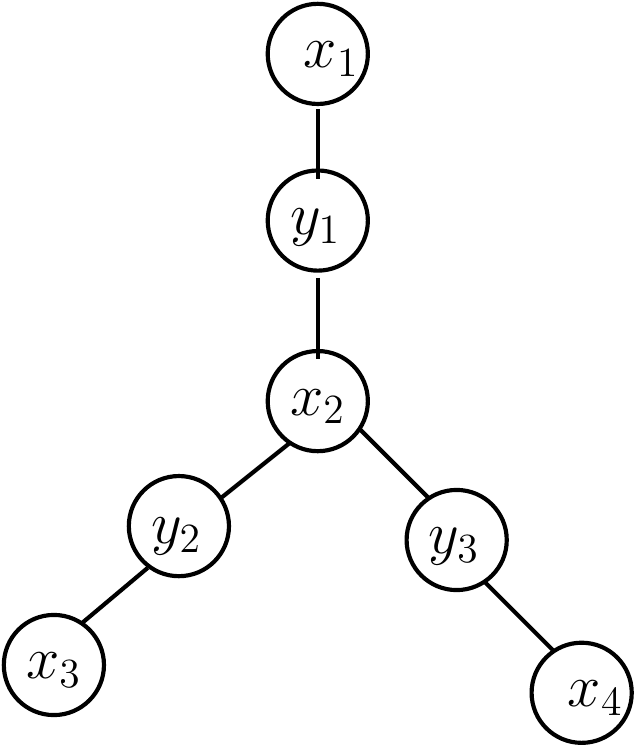}
\end{minipage}
 \begin{minipage}{.33\linewidth} 
  \[
\begin{blockarray}{cccccc}
   & & y_{1} & y_{2} & y_{3} & y_{4}\\
\begin{block}{cc(cccc)}
   x_{1} & & 1 & 1 & 0 & 0\\
   x_{2} & & 0 & 1 & 1 & 0 \\
   x_{3} & & 0 & 1 & 0 & 1 \\
   \end{block}
\end{blockarray}\]
  \end{minipage}
 \begin{minipage}{.2\linewidth}
\includegraphics[width=1.0 in]{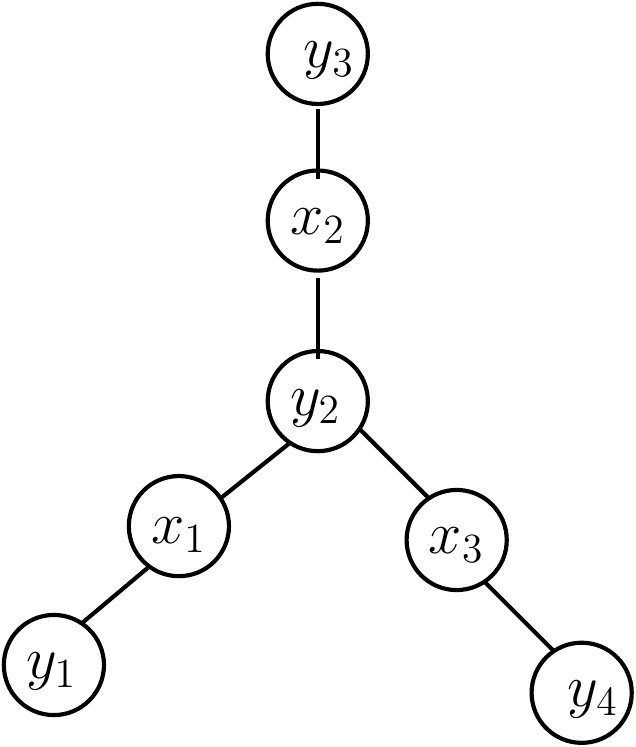}
\end{minipage}   
\indent \hspace{0.5 in} \textbf{(a)} \hspace{1.0 in} \textbf{(b)} \hspace{1.0 in} \textbf{(c)} \hspace{0.85 in} \textbf{(d)}
  \caption{(a) Forbidden submatrix ${M^{\mathsmaller T}_{\mathsmaller 3_{1}}}$ (b) ${G_{\mathsmaller{M}_{3_{1}}^{\mathsmaller T}}}$, the representing graph of ${M^{\mathsmaller T}_{\mathsmaller 3_{1}}}$ (c) Forbidden submatrix $M_{3_{1}}$ and (d) $G_{\mathsmaller{M}_{3_{1}}}$, the representing graph of $M_{3_{1}}$.  \label{eg1}}
\end{figure}
In this section, we present FPT algorithms for the problems $d$-$SC1S$-$R$, $d$-$SC1S$-$C$ and $d$-$SC1S$-$RC$ on $(2,*)$-matrices and $(*,2)$-matrices. Our algorithm makes use of the forbidden submatrix characterization for the $SC1P$ by Tucker (see Theorem \ref{thm2}). A similar technique is used in (\cite[Chapter 4]{dom2009recognition}) to prove the fixed-parameter tractability of $d$-$COS$-$R$ problem on $(2,*)$-matrices. We extend those results to $SC1S$ and $SC1E$ problems. Given an input matrix $M$, our algorithm consists of two stages. Stage $1$ first preprocess the input matrix to remove identical rows and columns and then destroys all fixed-size forbidden submatrices from $X$ in $M$. Stage $2$ focuses on destroying infinite-size forbidden submatrices in $M$.

\textit{Preprocessing}\label{preprocess} on the input matrix $M$ is done by assigning \textit{weights} to each row, column and entry and deleting all but one occurrence of identical rows and columns. For a matrix $M$, the \textit{weight of a row (resp. column)}\label{defn2} is equal to the number of times that row (resp. column) appears in $M$. The \textit{weight of an entry} is equal to the product of the weight of its row and column. Assigning weights to rows and columns ensures that preprocessing doesn't change the original matrix while deleting identical rows and columns. The resultant matrix thus obtained will have no identical rows and columns, and it is also possible for a matrix to have more than one row$/$column with equal weight. 

If $M$ is a $(2,*)/(*,2)$-matrix, then the only forbidden matrix from $X$ (Section \ref{defn1}) that can be appear in $M$ is  ${M^{\mathsmaller T}_{\mathsmaller 3_{1}}}/{M_{3_{1}}}$, because all other matrices in $X$ contain a column/row with more than two ones. We use a recursive branching algorithm, which is a search tree that checks for forbidden matrices of type ${M^{\mathsmaller T}_{\mathsmaller 3_{1}}}/{M_{3_{1}}}$ in $M$ and then branches recursively into three$/$four subcases, depending upon the problem under consideration. If the resultant matrix obtained after satge $1$ does not have the $SC1P$, then stage $2$ of our algorithm focuses on destroying the forbidden matrices of type $M_{I_{k}}$ and $M^{\mathsmaller T}_{\mathsmaller{I_{k}}}$ (where $k \geq 1)$ efficiently.

In stage $2$ of our algorithm, branching strategy cannot be applied to destroy $M_{I_{k}}$ and $M^{T}_{\mathsmaller{I_{k}}}$ (where $k \geq 1$), because their sizes are unbounded. We use the result of Theorem \ref{main} cleverly, to get rid of $M_{I_{k}}$ and $M^{T}_{\mathsmaller{I_{k}}}$ in stage $2$. 
\begin{figure}[t]
\centering
 \includegraphics[width=3.5 in]{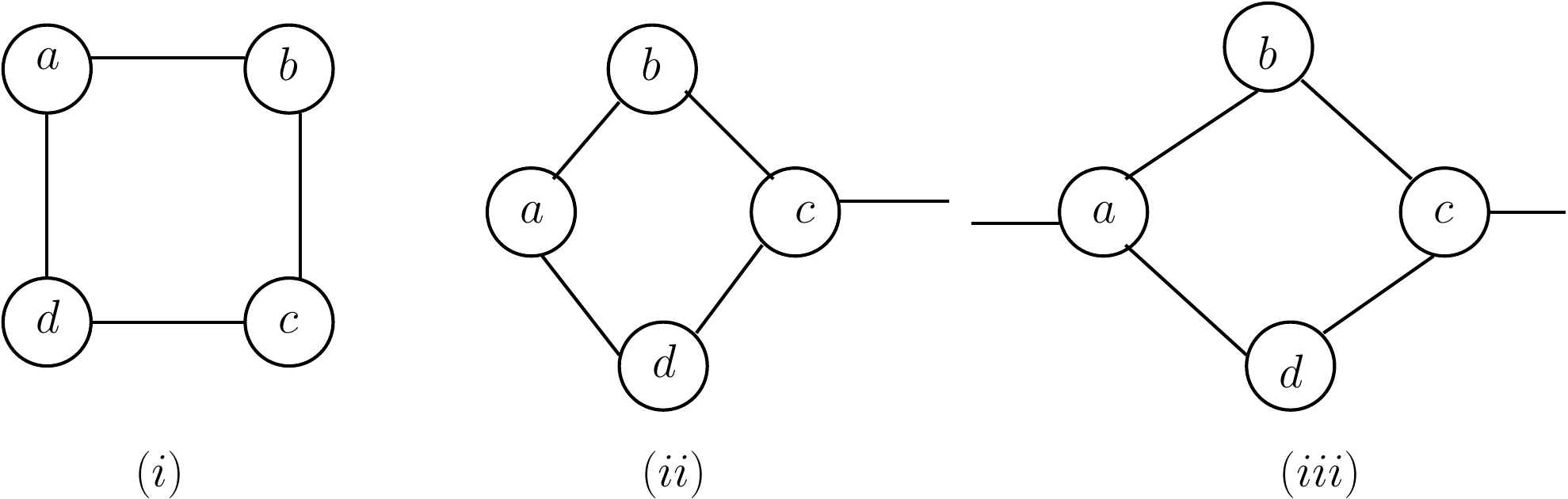}
  \caption{Possible chordless cycles of length four in a $(2,*)$-matrix and $(*,2)$-matrix .\label{chrd}}
\end{figure}
\begin{lemma}\label{chrdl}
If $M$ is a $(2,*)$-matrix or a  $(*,2)$-matrix, that does not have identical columns and identical rows, then there are no chordless cycles of length four in the representing graph $G_{\mathsmaller{M}}$, of $M$.
\end{lemma}
 \begin{proof}
 The possible chordless cycles of length four in the representing graph of a $(2,*)$ and $(*,2)$-matrices are shown in Figure \ref{chrd}. Here, we can note that the vertices $b$ and $d$ cannot have degree greater than two, because we are considering only $(2,*)$-matrices and  $(*,2)$-matrices. In Figure \ref{chrd} (i), (ii) and (iii), we can see that the vertices $b$ and $d$ are connected to the same vertices. That means the rows (or columns) corresponding to vertices $b$ and $d$ in $M$ are identical, which is a contradiction.
\end{proof}
\begin{theorem} \label{main}
Let $M$ be a $(2,*)$-matrix or $(*,2)$-matrix that does not have identical columns and identical rows. If $M$ does not have the $SC1P$  and does not contain matrices in $X$ as submatrices, then the matrices of type $M_{\mathsmaller I_{k}}$ and $M^{T}_{\mathsmaller{I_{k}}}$ (where $k \geq 1$), that are contained in $M$ are pairwise disjoint, i.e. they have no common column or row.
\end{theorem}
\begin{proof}
Consider the representing graph $G_{\mathsmaller{M}}$ of a given $m \times n$ matrix $M$. From Lemma \ref{chrdl}, it is clear that there are no chordless cycles of length $4$ in $G_{\mathsmaller{M}}$, since $M$ is a $(2,*)/(*,2)$-matrix with no identical rows and columns,. For a contradiction, assume that $M$ contains a pair of matrices of type $M_{I_{k}}$ and/or $M^{T}_{\mathsmaller{I_{k}}}$ (where $k \geq 1$), that share at least one common column or row. This implies that, there are two induced cycles of length at least six in $G_{\mathsmaller{M}}$, that have at least one vertex in common corresponding to a column or row of $M$ (Lemma \ref{prop3}). Figure \ref{chrd2} (a) and (b) shows the minimal possibilities for $G_{\mathsmaller{M}}$ to have two chordless cycles of length six that share at least one vertex. Each of these graphs have either a $G_{\mathsmaller{M}_{3_{1}}}$ or ${G_{\mathsmaller{M}_{3_{1}}^{\tiny T}}}$ (See Figure \ref{eg1} and \ref{chrd2})  as an induced subgraph. This means that, $M$ contains an $M_{3_{1}}$ or ${M^{T}_{\mathsmaller 3_{1}}}$, which is a contradiction, to the fact that all forbidden matrices in $X$ have been removed from $M$. The same can be proved by induction on chordless cycles of length eight, ten, twelve,\ldots,$2(min(m,n))$. Therefore our  assumption that two chordless cycles in $M$ share at least one vertex is wrong. Therefore, matrices of type $M_{\mathsmaller I_{k}}$ and $M^{T}_{\mathsmaller{I_{k}}}$ (where $k \geq 1$) that are contained in $M$ are pairwise disjoint.  
\end{proof}
\begin{figure}[t]
\centering
 \includegraphics[width=4.0 in]{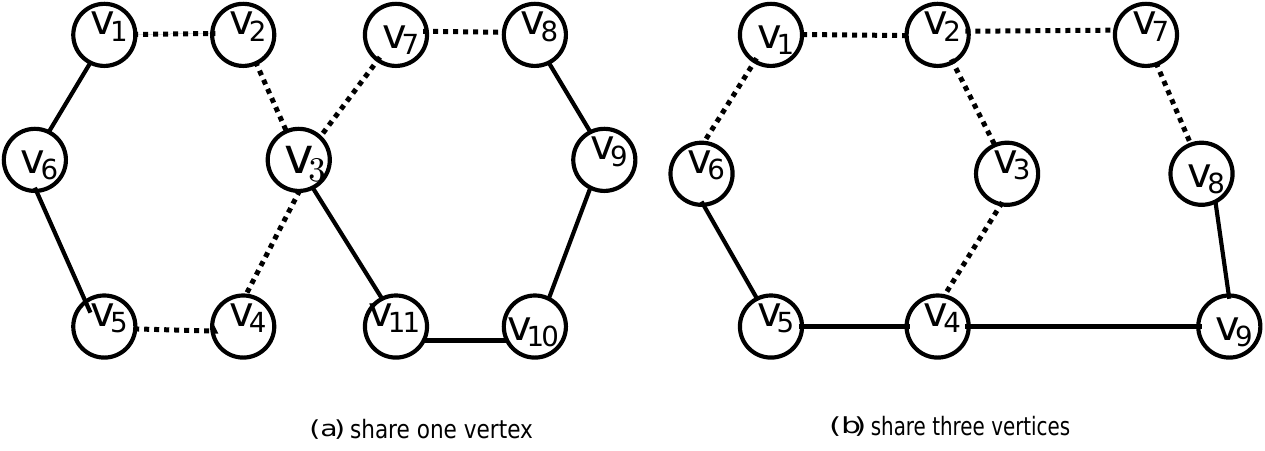}
  \caption{Minimal possibilities for the representing graph $G_{\mathsmaller{M}}$, of a $(2,*)$-matrix or $(*,2)$-matrix  to have two chordless cycles of length six that share at least one vertex. The edges shown in dotted lines are the edges of the the forbidden subgraphs $G_{\mathsmaller{M}_{3_{1}}}$ or ${G_{\mathsmaller{M}_{3_{1}}^{\tiny T}}}$. \label{chrd2}}
\end{figure}

\noindent \textbf{An FPT Algorithm for $d$-$SC1S$-$R$ : \newline}
In Algorithm \ref{alg1}, we present an FPT algorithm \textit{$d$-$SC1S$-row-deletion-restricted-matrices} for solving $d$-$SC1S$-$R$ problem on $(2,*)$-matrices. Given a matrix $M$ and a parameter $d$ (maximum number of rows that can be deleted), Algorithm \ref{alg1} first preprocess (Section \ref{preprocess}) the input matrix, and then search and destroy every submatrix of $M$ that contains an ${M^{\mathsmaller T}_{\mathsmaller 3_{1}}}$. If $M$ contains an ${M^{\mathsmaller T}_{\mathsmaller 3_{1}}}$, then the algorithm branches into at most four subcases (depending on the rows of ${M^{\mathsmaller T}_{\mathsmaller 3_{1}}}$ found in $M$). Each branch corresponds to deleting a row of the forbidden matrix ${M^{\mathsmaller T}_{\mathsmaller 3_{1}}}$ found in $M$. In each of the subcases, when a row is deleted, the parameter $d$ is decremented by the weight (Section \ref{defn2}) of that row. As long as $d > 0$, the above steps are repeated for each subcase until all the forbidden matrices of type ${M^{\mathsmaller T}_{\mathsmaller 3_{1}}}$ are destroyed. The number of leaf instances is at most $O(4^{d})$. For each of the leaf instances $M_{i}$, if the resulting matrix still does not have the $SC1P$, then the only possible forbidden matrices that can remain in $M_{i}$ are of type $M_{I_{k}}$ and $M^{\mathsmaller T}_{\mathsmaller{I_{k}}}$ (where $k \geq 1$). If they appear in $M_{i}$, by Theorem \ref{main} they are pairwise disjoint. Pairwise disjoint $M_{I_{k}}$ and $M^{\mathsmaller T}_{\mathsmaller{I_{k}}}$ in $M_{i}$, can be destroyed by deleting a row with minimum weight (by breaking ties arbitrarily) from each of them. On deletion of a row, the parameter $d$ is decremented by the weight of that row. If the sum of the weights of all the deleted rows is less than or equal to $d$ then, the algorithm returns Yes indicating that input is an Yes instance. Otherwise, the algorithm returns No.
\begin{algorithm}[h]
\caption{Algorithm  \textit{$d$-$SC1S$-row-deletion-restricted-matrices}$(M,d)$\label{alg1}}
\begin{algorithmic}[1]
 \Require An instance $\langle M_{m \times n},d \rangle$ where $M$ is a $(2,*)$-binary matrix and $d \geq 0$ \vspace{-0.1 in} \[ \hspace{-0.48 in}\textbf{Output} = \begin{dcases*} Yes,  & if there exists a set  $R^{'} \subseteq R(M)$, with $\vert R^{'} \vert \leq d$, such that $M \backslash R'$ \\  
        & has the $SC1P$.  
   \\  
   No,    & otherwise
\end{dcases*}
\] \vspace{-0.05 in} 
 \noindent \textit{\textbf{Stage 1:}}\vspace{-0.05 in} 
 \State Apply preprocessing steps as discussed in Section \ref{restrctd} on $M$. \vspace{-0.05 in} 
  \State \textbf{if}  {$M$ has the $SC1P$ and $d \geq 0$} \textbf{then},
 return Yes.\vspace{-0.05 in} 
\State \textbf{if} {$d < 0$} \textbf{then}, return No. \vspace{-0.05 in}  \newline \vspace{-0.05 in} 
\noindent \textit{\textbf{Branching Step:}} \State \textbf{if} {there exists a submatrix $M^{'}$ in $M$ that is isomorphic to ${M^{\mathsmaller T}_{\mathsmaller 3_{1}}}$} \textbf{then}, 

 \vspace{-0.05 in}
  Branch into at most $four$  instances $I_{i}=\langle M_{i}, d_{i}\rangle$ where $i \in \{1,2,3,4\}$ \vspace{-0.05 in} 
 
 Set $M_{i} \leftarrow M\setminus \{r_{i}\}$ \hspace{1.5 in}// $r_{i}$ denotes the $i^{th}$ row of $M^{'}$.   \vspace{-0.05 in}

 Update $d_{i} \leftarrow d$-$wt(r_{i})$  \hspace{0.90 in}// $wt(r_{i})$ denotes the weight of row $r_{i}$. \vspace{-0.05 in} 
 
\noindent For some $i \in \{1,2,3,4\}$, if \textit{$d$-$SC1S$-row-deletion-restricted-matrices}$(M_{i},d_{i})$ returns Yes, then return Yes, else if all instances return No, then return No. \vspace{-0.05 in}  
\State \textbf{end if}\vspace{-0.05 in} \newline \vspace{-0.05 in} 
\noindent \textit{\textbf{Stage 2:}} \vspace{-0.02 in}  
\State \textbf{while} \footnotesize there exists a  submatrix $N$ in $M$ that is isomorphic to an $M_{I_{k}}$ or $M^{T}_{\mathsmaller{I_{k}}}$ and $d>0$ 
 \textbf{do},\vspace{-0.05 in} \normalsize
\State \indent Delete a row $r$ in $N$ having minimum weight from $M$.\vspace{-0.05 in} 
\State \indent Decrement the parameter $d$ by the weight of the deleted row $r$. \vspace{-0.05 in} 
\State \textbf{end while}\vspace{-0.05 in} 
\State \textbf{if} {$M$ does not contain $M_{I_{k}}$ or $M^{T}_{\mathsmaller{I_{k}}}$ and $d \geq 0$} \textbf{then} return Yes, otherwise return No\vspace{-0.06 in} 
\end{algorithmic}
\end{algorithm}
\newline \noindent The correctness of the branching step can be explained in the following Lemma.
 \begin{lemma}\label{corr1}
Let $M$ be a $(2,*)$-matrix that does not have the $SC1P$. Suppose $M$ contains one of the forbidden matrices from $X$. Let $M[\{r_{1},\ldots,r_{4}\}]$ be a submatrix that contains a forbidden matrix from $X$, where $\{r_{1},\ldots,r_{4}\}  \subseteq R(M)$. Then, any solution of $d$-$SC1S$-$R$ includes at least one of the rows $r_{1},\ldots,r_{4}$.
 \end{lemma}
 \begin{proof}
 Assume that there exists a solution for $d$-$SC1S$-$R$, say $S$ that contains none of the rows $r_{1},r_{2},\ldots,r_{4}$. Let $M^{'} = M \backslash S$ be the matrix with the $SC1P$. This implies that $M[\{r_{1},r_{2},\ldots,r_{4}\}]$ in $M^{'}$ satisfies the $SC1P$, which is a contradiction.
 \end{proof}
  Algorithm \ref{alg1} can be used to solve $d$-$SC1S$-$R$ problem on $(*,2)$-matrices also, by searching for an $M_{3_{1}}$ instead of ${M^{\mathsmaller T}_{\mathsmaller 3_{1}}}$ in $M$ (in line $6$ of Algorithm \ref{alg1}), and considering the number of branches as three (since the only forbidden matrix in $X$ that can occur in a $(*,2)$-matrix is $M_{3_{1}}$ and it has three rows). 
 \begin{theorem}\label{rest4}
$d$-$SC1S$-$R$ is fixed-parameter tractable on  $(2,*)$/$(*,2)$-matrices with a run-time of $O^{*}(4^{d})$/$O^{*}(3^{d})$, where $d$ denotes the number of rows that can be deleted.
\end{theorem}
\begin{proof}
Algorithm \ref{alg1} employs a search tree, where each node in the search tree has at most four/three subproblems, and therefore the tree has at most $4^{d}$/$3^{d}$ leaves. The size of the search tree is $O(4^{d})$/$O(3^{d})$. A submatrix $M^{'}$ of $M$, that is isomorphic to ${M^{T}_{\mathsmaller 3_{1}}}$ and $M_{3_{1}}$ can be found in $O(m^{4}n)$-time and $O(m^{3}n)$-time(using Lemma \ref{prop1}) respectively. Therefore, for an input matrix $M$, the time required for destroying an ${M^{T}_{\mathsmaller 3_{1}}}/M_{3_{1}}$ (stage 1) is $O(4^{d}m^{4}n)$/$O(3^{d}m^{3}n)$. The time required for finding a submatrix of type $M_{I_{k}}$ and $M^{T}_{\mathsmaller{I_{k}}}$, (where $k \geq 1$) in $M$ is $O(m^{3}n^{3})$ and $O(m^{3})$ (using Lemma \ref{prop2}) on $(2,*)$-matrices and $(*,2)$-matrices respectively. For each of the leaf instance $M_{i}$, line $6$ of Algorithm \ref{alg1} is executed at most $d_{i}$ times and $d_{i} \leq d$. Therefore the time complexity of destroying $M_{I_{k}}$ and $M^{T}_{\mathsmaller{I_{k}}}$ in $M$ (stage 2) is $O(4^{d}m^{3}n^{3}d)$/$O(3^{d}m^{3}d)$. The total time complexity of Algorithm \ref{alg1} is $O(4^{d}(m^{4}n + d.m^{3}n^{3}))$/$O(3^{d}(m^{3}n + d.m^{3}))$ on $(2,*)$/$(*,2)$-matrices. \end{proof}
The following corollary on \textsc{Proper Interval Vertex Deletion} problem (Section \ref{proper}) is a direct consequence of Theorem \ref{rest4}.
\begin{corollary}
\textsc{Proper Interval Vertex Deletion} problem is fixed-parameter tractable on triangle-free graphs with a run-time of $O^{*}(4^{d})$, where $d$ denotes the number of allowed vertex deletions.
\end{corollary}
 \begin{corollary}\label{cor5}
The optimization version of $d$-$SC1S$-$R$ problem \textup{(\textsc{Sc$\mathsmaller1$s-Row Deletion})} on a $(2,*)/(*,2)$-matrix can be approximated in polynomial-time with a factor of four/three.
 \end{corollary}
 \begin{proof}
 In Stage $1$ of Algorithm \ref{alg1}, instead of branching on each of the rows of a forbidden submatrix ${M^{T}_{\mathsmaller 3_{1}}}$/$M_{3_{1}}$ found in $M$, delete all rows of each of the forbidden submatrix ${M^{T}_{\mathsmaller 3_{1}}}$/$M_{3_{1}}$ found in $M$. From Algorithm \ref{alg1}, it is clear that Stage $2$ solves the problem exactly. This results in a $4$-factor/$3$-factor approximation algorithm. 
 \end{proof}
 
\noindent \textbf{An FPT Algorithm for $d$-$SC1S$-$C$: \newline}
A related problem of deleting at most $d$ number of columns to get the $SC1P$ ($d$-$SC1S$-$C$ problem) can also be solved using Algorithm \ref{alg1} (consider the columns instead of rows in lines 4, 7 and 8) in $O^{*}(3^{d})$-time for $(2,*)$-matrix ($O^{*}(4^{d})$-time for $(*,2)$-matrix) and the approximation factor for the optimization version of $d$-$SC1S$-$C$ problem \textup{(\textsc{Sc$\mathsmaller 1$s-Column Deletion})} is three (four).\\

\noindent \textbf{An FPT Algorithm for $d$-$SC1S$-$RC$: \newline}
 $d$-$SC1S$-$RC$ problem can also be solved using Algorithm \ref{alg1} (consider the rows as well as columns instead of rows in lines 4, 7 and 8) in $O^{*}(7^{d})$-time on $(2,*)/ (*,2)$-matrices. The approximation factor for the optimization version of $d$-$SC1S$-$RC$ problem \textup{(\textsc{Sc$\mathsmaller 1$s-Row-Column Deletion})} is seven.
\subsection{\textup{\textbf{Establishing $SC1P$ by Flipping Entries}}} \label{etrowcolumn}
This section considers $SC1E$ problems by flipping $0/1$ and $01$-entries. We refer to the decision versions of the optimization problems $SC1P$-$0$-$Flipping$ and $SC1P$-$01$-$Flipping$ defined in Section \ref{opt} as $k$-$SC1P$-$0E$ and $k$-$SC1P$-$01E$ respectively, where $k$ denotes the number of allowed flippings. First, we show that these problems are NP-complete. Then, we give an FPT algorithm for $d$-$SC1P$-$0E$ problem on general matrices. Finally, we present an FPT algorithm for $d$-$SC1P$-$1E$ problem on certain restricted matrices. 
\subsubsection{\textup{\textbf{NP-completeness}}}\label{hard2}
The following theorem proves the NP-completeness of $k$-$SC1P$-$0E$ problem using \textsc{$k$-Chain Completion} problem (Definition \ref{chain}) on bipartite graphs as a candidate problem. 
 \begin{theorem} \label{thm4}
 The $k$-$SC1P$-$0E$ problem is NP-complete.
 \end{theorem}
 \begin{proof}
 We first show that $k$-$SC1P$-$0E$ $\in NP$. Given a matrix $M$ and an integer $k$, the certificate is a set $A'$ of indices corresponding to $0$-entries in $M$. The verification algorithm affirms that $\vert A' \vert \leq k$, and then it checks whether flipping these $0$-entries in $M$ yields a matrix with the $SC1P$. This verification can be done in polynomial time. 
 
 We prove that $k$-$SC1P$-$0E$ problem is NP-hard by showing that \textsc{$k$-Chain Completion} $\leq_{\mathsmaller {P}}$ $k$-$SC1P$-$0E$. The half-adjacency matrix of any chain graph can be observed to satisfy the $SC1P$, however the converse is not true. Given a bipartite graph $G=(V_{\mathsmaller{1}}, V_{\mathsmaller{2}}, E)$ with $V_{\mathsmaller{1}}=m$ and $V_{\mathsmaller{2}}=n$, we create a $2m \times 2n$ binary matrix $M_{G_{new}}$ as follows. $M_{G_{new}}= \begin{bmatrix}J_{m,n} &M_\mathsmaller{G}\\0_{m,n} &J_{m,n}\end{bmatrix} $ = $\begin{bmatrix}A &B \\D &C\end{bmatrix}$, where $M_{\mathsmaller{G}}$ is the half adjacency matrix of $G$, $J_{m,n}$ is an $m \times n$ matrix with all entries as one and $0_{m,n}$ is an $m \times n$ matrix with all entries as zero. It can be noted that adding an edge in $G$ corresponds to flipping a $0$-entry in $B$. We show that  $G$ can be converted to a chain graph $G'$ by adding at most $k$ edges if and only if there are at most $k$ number of $0$-flippings in $M_{G_{new}}$, such that the resultant matrix $M_{G'_{new}}$ satisfies the $SC1P$.

\indent Suppose $G'$ is a chain graph, then $\begin{bmatrix} 1 &0\\ 0 &1
\end{bmatrix}$ cannot occur exclusively in $B$ (from Lemma \ref{chnrs}). By construction of $M_{G_{new}}$, it can be observed that $\begin{bmatrix} 1 &0 &0\\ 0 &1 &0 \end{bmatrix} $ and $\begin{bmatrix} 1 &0\\ 0 &1\\ 0 &0 \end{bmatrix}$ cannot occur as submatrices in $M_{G'_{new}}$. From Figure \ref{forbidden}, it is clear that one of the configurations of these two matrices occur as a submatrix in all the forbidden submatrices of the SC1P, except $M_{I_{1}}$. Hence $M_{I_1}$ is the only forbidden submatrix of the $SC1P$ that could appear in $M_{G'_{new}}$. However, if $M_{G'_{new}}$ contains $M_{I_1}$, then it would further imply that $B'$(matrix obtained after flipping the 0-entries of $B$) contains $\begin{bmatrix} 1 &0\\ 0 &1 \end{bmatrix}$ as a submatrix, which contradicts the assumption that $G'$ is a chain graph. Therefore, if $k$ edges can be added to $G$ to make it a chain graph, then $k$ $0$-entries can be flipped in $M_{G_{new}}$ to make it satisfy the $SC1P$.

Conversely, suppose that $k$=$k_1$+$k_2$ $0$-flippings are performed on $M_{G_{new}}$ to make it satisfy the $SC1P$, where $k_1$ and $k_2$ refer to the number of $0$-flippings performed in $B$ and $D$ respectively. Let us assume that the corresponding bipartite graph $G'$, obtained after the flipping of zeroes in $B$ is not a chain graph. Since $G'$ is not a chain graph, it contains $2K_2$ as an induced subgraph, which further means that $B'$ contains $\begin{bmatrix}1 &0\\0 &1 \end{bmatrix} $ as a submatrix. The construction of $M_{G_{new}}$ implies that $M_{G'_{new}}$ has $M_{I_1}$ as a submatrix (considering the remaining $3$ quadrants of $M_{G'_{new}}$), which leads to a contradiction. Hence $G'$ is a chain graph. Therefore, $k$-$SC1P$-$0E$ is NP-complete.
\end{proof}
The following corollary on \textsc{Biconvex Completion} problem is a direct consequence of the above theorem.
\begin{corollary}
Given a bipartite graph $G=(V_{1},V_{2},E)$ and a non-negative integer $k$, the problem of deciding whether there exists a set $E' \subseteq (V_{1} \times V_{2})\backslash E$, of size at most $k$, such that $G=(V_{1},V_{2},E \cup E')$ is a biconvex graph is NP-complete.
\end{corollary}
The following theorem proves the NP-completeness of the $k$-$SC1P$-$01E$ problem using the \textsc{$k$-Chain Editing} problem (Definition \ref{chain}) on bipartite graphs as a candidate problem. 
 \begin{theorem}
 The $k$-$SC1P$-$01E$ problem is NP-complete.
 \end{theorem}
 \begin{proof}
We first show that $k$-$SC1P$-$01E$ $ \in NP$. Given a matrix $M$ and an integer
$k$, the certificate is a set $A'$ of indices corresponding to $0/1$-entries in $M$. The verification algorithm affirms that $ \vert A' \vert \leq k$, and then it checks whether flipping these $0/1$-entries in $M$ yields a matrix with the SC1P. This verification can be done in polynomial time.

\indent We prove that $k$-$SC1P$-$01E$ is NP-hard by showing that \textsc{$k$-Chain Editing} $\leq_{\mathsmaller {P}}$ $k$-$SC1P$-$01E$. The NP-hardness of $k$-$SC1P$-$01E$ can be proved similar to the NP-hardness of $k$-$SC1P$-$0E$ (Theorem \ref{thm4}) by considering $M_{G_{new}}$ as follows: $M_{G_{new}}$ = $\begin{bmatrix}J_{m,mn} &M_\mathsmaller{G} \\0_{mn,mn} &J_{mn,n}\end{bmatrix} $, where $G=(P,Q,E)$ is a bipartite graph, with $\lvert$P$\lvert$=$m$ and $\lvert$Q$\lvert$=$n$ and $M_\mathsmaller{G}$ being the half adjacency matrix of $G$. It can be noted that adding/removing an edge in $G$ corresponds to flipping a $0/1$-entry in $B$. We claim that $G$ can be converted to a chain graph $G'$ by adding/deleting at most $k$ edges if and only if there are at most $k$ number of $0/1$-flippings in $M_{G_{new}}$, such that the resultant matrix $M_{G'_{new}}$ satisfies the $SC1P$ (This can be proved similar to Theorem \ref{thm4}). 
\end{proof}
The following corollary on \textsc{Biconvex Editing} problem is a direct consequence of the above theorem.
\begin{corollary}
Given a bipartite graph $G=(V_{1},V_{2},E)$ and a non-negative integer $k$, the problem of deciding whether there exists a set of at most $k$ edge modifications (edge additions/deletions) in $G$, that results in a biconvex graph is NP-complete.
\end{corollary}
\subsubsection{\textup{\textbf{An FPT algorithm for $d$-$SC1P$-$0E$ problem}}} \label{fpt0e}
In this section, we present an FPT algorithm \textit{$d$-$SC1P$-$0$-Flipping} (Algorithm \ref{alg3}), for $d$-$SC1P$-$0E$ problem on general matrices. Given a binary matrix $M$ and a non-negative integer $d$, Algorithm \ref{alg3} destroys forbidden submatrices from $F_{SC1P}$ in $M$, using a simple search tree based branching algorithm. The algorithm recursively branches, if $M$ contains a forbidden matrix from $X$ (see Section \ref{defn1}) as well as $M_{\mathsmaller I_{k}}$ or $M^{T}_{\mathsmaller I_{k}}$ (where $k \geq 1)$. If $M$ contains a forbidden matrix from $X$, then the algorithm branches into at most eighteen subcases, since the largest forbidden matrix of $X$ has eighteen $0$-entries. In each subcase, flip one of the $0$-entry of the forbidden submatrix found in $M$ and decrement the parameter $d$ by one. Otherwise, if $M$ contains a forbidden submatrix of type $M_{\mathsmaller I_{k}}$ or $M^{T}_{\mathsmaller I_{k}}$, then the algorithm finds a minimum size forbidden matrix $M'$, of type $M_{\mathsmaller I_{k}}$ or $M^{T}_{\mathsmaller I_{k}}$ in $M$. If the value of $k$ is greater than $d$, then the algorithm returns No (using Corollary \ref{prp3}), otherwise the algorithm branches into at most $O(7^{d})$-subcases (using Lemma \ref{obs2}). In each subcase, flip $k$ $0$-entries of the forbidden submatrix $M'$ found in $M$, and decrement the parameter $d$ by $k$. This process is continued in each subcase, until its $d$ value becomes zero or until it satisfies the $SC1P$. Algorithm \ref{alg3} returns Yes if any of the subcases returns Yes, otherwise it returns No. 

Flipping a $0$-entry in $M$ is equivalent to adding an edge in the representing graph $G_{\mathsmaller{M}}$ of $M$. From this fact and Lemma \ref{prop3}, it follows that to destroy $M_{\mathsmaller I_{k}}$ and $M^{T}_{\mathsmaller I_{k}}$ in $M$, it is sufficient to destroy chordless cycles of length greater than four in $G_{\mathsmaller{M}}$ (i.e make $G_{\mathsmaller{M}}$ a chordal bipartite graph (Section \ref{prel}) by addition of edges). The number of zero flippings required to destroy an $M_{\mathsmaller I_{k}}$ or $M^{T}_{\mathsmaller I_{k}}$, where $(k \geq 1)$ is given in Corollary \ref{prp3}.
\begin{corollary}\label{prp3}
The minimum number of $0$-flippings required to destroy an $M_{I_{k}}$ or $M^{T}_{\mathsmaller I_{k}}$, where $(k \geq 1)$ is $k$.
\end{corollary}
\begin{proof}
It follows from Lemma \ref{main1} and Lemma \ref{prop3}.
\end{proof}

\begin{algorithm}[h]
\caption{Algorithm  \textit{$d$-$SC1P$-$0$-Flipping}$(M,d)$\label{alg3}}
\begin{algorithmic}[1] 
 \Require An instance $\langle M_{m \times n},d \rangle$, where $M$ is a binary matrix and $d \geq 0$.\vspace{-0.1 in} \[ \hspace{-0.26 in}\textbf{Output} = \begin{dcases*} Yes,  & if there exists a set $B^{'} \subseteq B(M)$, of indices with $\vert B^{'} \vert \leq d$, such that the resultant \\  
        & matrix obtained by  flipping the entries of  $B^{'}$ satisfy the $SC1P$.  
   \\  
   No,    & otherwise
\end{dcases*}
\] \vspace{-0.2 in} 
 \State \textbf{if} {$M$ has the $SC1P$ and $d \geq 0$} \textbf{then} return Yes. \vspace{-0.1 in}
 \State \textbf{if}  {$d < 0$} \textbf{then} return No.  \vspace{-0.05 in} \newline \vspace{-0.05 in}
\noindent \textit{\textbf{Branching Step:}}
 \vspace{0.03 in}
\State \textbf{if} {$M$ contains a forbidden submatrix $M'$ from $X$} \textbf{then},\vspace{-0.08 in}
\State  \indent Branch into at most $18$  instances $I_{i}=\langle M_{i}, d_{i}\rangle$ where $i \in \{1,2,\ldots,18\}$ \vspace{-0.08 in} \newline \vspace{-0.08 in} 
 \indent Set $M_{i} \leftarrow M$ with $i^{th}$ $0$-entry of $M^{'}$ flipped (where $0$-entries of $M'$ are labelled \vspace{-0.05 in} \indent \indent \indent \indent \indent \indent \indent \indent \indent \indent \indent in row-major order). \vspace{-0.03 in} \vspace{-0.03 in} \newline  \vspace{-0.03 in}
 \indent Update $d_{i} \leftarrow d -1$  \hspace{1.0 in}// Decrement parameter by 1.
\vspace{-0.03 in} \newline  \vspace{-0.03 in}
\noindent For some $i \in \{1,2,\ldots,18\}$, if \textit{$d$-$SC1P$-$0$-Flipping}$(M_{i},d_{i})$ returns Yes, then \vspace{-0.05 in} \newline \vspace{-0.08 in} return Yes, else if all instances return No, then return No.  
\vspace{0.05 in}
\State  \textbf{if} {$M$ contains either $M_{I_{k}}$ or $M^{T}_{\mathsmaller{I_{k}}}$} \textbf{then},\vspace{-0.08 in}
\State  \indent  Find a minimum size $M_{I_{k}}$ or $M_{I_{k}}^{T}$ in $M$, (say $M'$) \vspace{-0.08 in}
\State  \indent \textbf{if} {$k > d$}, return No.\vspace{-0.08 in}
\State  \indent \textbf{else}\vspace{-0.08 in}
 \State  \indent \indent Branch into at most $O(7^{k})$ (number of ways to destroy $M'$) instances 
 
 \vspace{-0.08 in}  \indent \indent \indent \indent \indent \indent \indent \indent $I_{i}=\langle {M_{i}}, d_{i}\rangle$ where $i \in \{1,2,\ldots, 7^{k}\}$. \vspace{-0.05 in}\vspace{-0.05 in}
\State  \indent \indent Set $M_{i} \leftarrow M$ with $k$ appropriate $0$-entries of $M'$ flipped.  \vspace{-0.05 in}
 \State  \indent \indent Update $d_{i} \leftarrow d-k$  \hspace{1.0 in}// Decrement parameter by $k$.\vspace{-0.08 in}
\State  \indent \textbf{end if} \vspace{-0.08 in}
\State  \textbf{end if} \vspace{-0.08 in}
\newline \vspace{-0.08 in}
\hspace{-0.15 in} For some $i \in \{1,2,\ldots,O(7^{k})\}$, if \textit{$d$-$SC1P$-$0$-Flipping}$(M_{i},d_{i})$ returns  Yes, then return Yes, else if all instances return No, then return No.
\end{algorithmic}
\end{algorithm}

\begin{Observation}
The number of $0$-entries in an $M_{I_{k}}$ or $M^{T}_{\mathsmaller I_{k}}$, where $(k \geq 1)$ is $O(k^{2})$.
\end{Observation}
The above observation leads to a $O^{*}(d^{2d})$ algorithm for $d$-$SC1P$-$0E$. But, using the result of the following lemma, we get a $O^{*}(18^{d})$ algorithm for $d$-$SC1P$-$0E$.
\begin{lemma}\label{obs2}
Given a bipartite graph $H=(V_{\mathsmaller{1}}, V_{\mathsmaller{2}}, E)$, which is an even chordless cycle of length $2n$ (where $n \geq 3$), the number of ways to make $H$ a chordal bipartite graph by adding $n$-$2$ edges is at most $6.75^{n-1}$.
\end{lemma}
\begin{proof}
Number of ways to make $H$ a chordal bipartite graph = Number of ternary trees with $n$-$1$ internal nodes (using Lemma \ref{main1}).

Number of ternary trees with $n$ internal nodes = $ \dfrac{{\binom {3n+1} {n}}}{3n+1}$
 
\hspace{1.5 in}= $\dfrac{{\binom {3n} {n}}}{2n+1}$ = $\dfrac{(3n)!}{(2n+1)(2n)!n!}$

$\lim_{n\to\infty} n! = {\sqrt{2\pi{n}}{(\dfrac{n}{e})}^{n}}$ (using Lemma \ref{main2}).

$\lim_{n\to\infty} \dfrac{{\binom {3n} {n}}}{2n+1} = \dfrac{{\sqrt{2\pi (3n)}{(\dfrac{3n}{e})}^{3n}}} {{\sqrt{2\pi (2n)}{(\dfrac{2n}{e})}^{2n}}\times {\sqrt{2\pi (n)}{(\dfrac{n}{e})}^{n}} \times (2n+1)}$

\hspace{1.5 in} $ = \dfrac{ {\sqrt{3}}\times3^{3n} }{ {\sqrt{4\pi{n}}}\times2^{2n}\times(2n+1) } $

$ \dfrac{{\binom {3n} {n}}}{2n+1} = O({\dfrac{3^{3n}}{\sqrt{n}\times2^{2n}\times(2n+1)}}) \sim O({\dfrac{3^{3n}}{2^{2n}}}) = O(6.75^n) $\\

Therefore, number of ternary trees with $n$ internal nodes = $O(6.75^n)$.
 
Hence, the number of ways to make $H$ a chordal bipartite graph is same as the number of ternary trees with $n$-$1$ internal nodes and is $O(6.75^{n-1})$. 
\end{proof}
\begin{lemma}\label{lem2}
In Algorithm \ref{alg3}, destroying $M_{\mathsmaller I_{k}}/M_{\mathsmaller I_{k}}^{T}$ takes $O^{*}(6.75^{d})$-time.
\end{lemma}
\begin{proof}
Let $\phi(k)$ represent the number of $4$-cycle decompositions (Section \ref{4-cycle}) of a $2(k + 2)$-cycle, or rather the representing graph of an $M_{1_{k}}$ or $M_{1_{k}}^{\mathsmaller T}$, where $k \geq 1$. Using Lemma \ref{tern}, $\phi(k)= \frac {\binom {3k+3} {k+1}} {2k+3}$. Let there be chordless cycles of sizes $2(i_{1}+2)$, $2(i_{2}+2)$, \ldots, $2(i_{m}+2)$ (in the non-decreasing order of size) in the representing graph of the input matrix, and let $d$ be the number of allowed edge additions. Since $i_{n}$ is equal to the number of edges to be be added to the $n^{th}$ smallest cycle in the representing graph of the input matrix, we get $\sum_{n=1}^{m} i_{n} \leq d$. Then, the number of leaves associated with the removal of $M_{I_{k}}/M_{I_{k}}^{T}$ in the search tree is given by:

$T(k)=\phi(i_{1}) \times \phi(i_{2}) \times \phi(i_{3}), \ldots \phi(i_{k}) =\prod_{n=1}^{m} \phi(i_{n})$\\ 
\indent \indent \indent$\leq \prod_{n=1}^{m} \phi(i_{n}) \frac{(2i_{n}+3)(i_{n}+1)}{3i_{n}+1}$,  where $\sum_{n=1}^{m} i_{n} \leq d$.\\ \vspace{0.05 in}
\indent \indent \indent $= \prod_{n=1}^{m} \binom {3.i_{n}} {i_{n}}$ \\ \vspace{0.05 in}
\indent \indent \indent $=O(\frac{\sqrt{2\pi(3d)}(\frac{3d}{e})^{3d}}{\sqrt{2 \pi(2d)}(\frac{2d}{e})^{2d}.\sqrt{2 \pi d}(\frac{d}{e})^{d}})$ \hspace{0.8 in}(Using Lemma \ref{main2}).\\
\indent \indent \indent$=O(\frac {3^{3d}}{\sqrt{d}. 2^{2d}}) = O(\frac{3^{3d}}{2^{2d}})= O(6.75^{d})$.\\
Hence, destroying all $M_{\mathsmaller I_{k}}/M^{T}_{\mathsmaller I_{k}}$ (where $k \geq 1$), in $M$ takes $O^{*}(6.75^{d})$-time.
\end{proof}
\begin{theorem} \label{rest2}
$d$-$SC1P$-$0E$ problem on a matrix $M_{m \times n}$, can be solved in $O^{*}(18^{d})$-time, where $d$ denotes the number of $0$-entries that can be flipped. Consequently, it is FPT.
\end{theorem}
\begin{proof}
Each node in the search tree of Algorithm \ref{alg3} has at most $18$ or $O(7^{k})$ subproblems, depending on whether we are destroying the fixed size forbidden matrices or $M_{\mathsmaller I_{k}}/M^{T}_{\mathsmaller I_{k}}$ respectively. A submatrix $M^{'}$ of $M$, that is isomorphic to one of the forbidden matrices in $X$, and $M_{I_{k}}/M^{T}_{\tiny I_{k}}$ can be found in $O(m^{6}n)$-time (using Lemma \ref{prop1}) and  $O(m^{3}n^{3})$-time (using Lemma \ref{prop2}) respectively. From Lemma \ref{lem2}, it follows that destroying all $M_{\mathsmaller I_{k}}/M^{T}_{\mathsmaller I_{k}}$ in $M$ takes $O^{*}(6.75^{d})$-time, whereas destroying all the forbidden matrices from $X$ takes $O^{*}(18^{d})$-time.  Therefore, the total time complexity of Algorithm \ref{alg3} is $O^{*}(18^{d})$ . 
\end{proof}
The idea used in Algorithm \ref{alg3}, does not work for $SC1S$ and other $SC1E$  problems defined in Section \ref{intro}. In $d$-$SC1P$-$0E$ problem, the presence of a large $M_{I_{k}}$ (or a large chordless cycle), where $k>d$ is enough to say that we are dealing with a No instance (Using Corollary \ref{prp3}). But for $d$-$SC1S$-$R \backslash C \backslash RC$ and $d$-$SC1P$-$1E\backslash01E$ problems, a chordless cycle (of any length) can be destroyed by deleting an arbitrary vertex and an arbitrary edge respectively. This idea plays a crucial role in the context of flipping zeroes, but not flipping ones, and deleting rows/columns in the input matrix.\\
\indent The following corollary on \textsc{Biconvex Completion} problem (Section \ref{intro}) is a direct consequence of Theorem \ref{rest2}.
\begin{corollary}
\textsc{Biconvex Completion} problem is fixed-parameter tractable on bipartite graphs with a run-time of $O^{*}(18^{d})$, where $d$ denotes the number of allowed edge additions.
\end{corollary}

\subsubsection{\textup{\textbf{An FPT algorithm for $d$-$SC1P$-$1E$ problem on restricted matrices}}}\label{restrctd1}
The $d$-$SC1P$-$1E$ problem on $(2,*)/(*,2)$-matrices can also be solved  using Algorithm \ref{alg1}, with a modification in the branching step as follows. Here, we branch on the number of $1$-entries of the forbidden submatrix $M^{T}_{\mathsmaller {3_{1}}}/M_{3_{1}}$  found in $M$. In each branch, we flip the corresponding $1$-entry and the parameter $d$ is decremented by the weight of that $1$-entry (Definition \ref{defn2}). The number of $1$-entries in an $M^{T}_{\mathsmaller{3_{1}}}/M_{3_{1}}$ is $6$ (for both $(2,*)$ and $(*,2)$-matrix), which leads to a branching factor of at most $6$. After the branching step, the remaining pairwise disjoint forbidden submatrices of type $M_{\mathsmaller I_{k}}$ and $M^{T}_{\mathsmaller{I_{k}}}$ (where $k \geq 1$) in $M$ can be destroyed in polynomial time by flipping a minimum weight $1$-entry in $M_{\mathsmaller I_{k}}$ and  $M^{T}_{\mathsmaller{I_{k}}}$ respectively. Therefore, the total time complexity is $O^{*}(6^{d})$, which leads to the following theorem.
\begin{theorem}\label{thm18}
$d$-$SC1P$-$1E$ on a $(2,*)/(*,2)$-matrix $M_{m \times n}$ can be solved in  $O^{*}(6^{d})$-time where $d$ denotes the number of allowed $1$-flippings. The optimization version of $d$-$SC1P$-$1E$ problem (\textsc{Sc$1$p-$1$-Flipping}) can be approximated in polynomial-time with a factor of six. 
\end{theorem}
The following corollary on \textsc{Biconvex Edge Deletion} problem (Section \ref{intro}) is a direct consequence of Theorem \ref{thm18}.
\begin{corollary}
\textsc{Biconvex Edge Deletion} problem is fixed-parameter tractable on certain bipartite graphs, in which the degree of all vertices in one partition is at most two, with a run-time of $O^{*}(6^{d})$, where $d$ denotes the number of allowed edge deletions.
\end{corollary}
\section{Conclusion}
In this work, first we showed that the decision versions of $SC1S$ and $SC1E$ problems are NP-complete. Then, we proved that $d$-$SC1S$-$R/C/RC$ and $d$-$SC1P$-$0E$ problems are fixed-parameter tractable on general matrices. We also showed that $d$-$SC1P$-$1E$ problem is fixed-parameter tractable on certain restricted matrices. Improved FPT algorithms for $d$-$SC1S$-$R/C/RC$ problems on $(2,*)$ and $(*,2)$ matrices are also presented here. We also observed that the fixed-parameter tractability of $d$-$SC1S$-$R$ problem on $(2,*)$-matrices implies that \textsc{Proper Interval Vertex Deletion} problem is FPT on triangle-free graphs with a run-time of $O^{*}(4^{d})$. From our results, it turns out that \textsc{Biconvex Vertex Deletion} and \textsc{Biconvex Completion} problems are fixed-parameter tractable. We also observed that \textsc{Biconvex Edge Deletion} problem is fixed-parameter tractable on certain restricted bipartite graphs. We conjecture that $d$-$SC1P$-$1E$ and $d$-$SC1P$-$01E$ problems are also fixed-parameter tractable on general matrices. However, the idea used for solving $d$-$SC1P$-$0E$ cannot be extended to solve $d$-$SC1P$-$1E$ and $d$-$SC1P$-$01E$ problems. In $d$-$SC1P$-$1E$ and $d$-$SC1P$-$01E$ problems, a chordless cycle of any length can be destroyed by removing a single edge, which leads to an unbounded number of branches. An interesting direction for future work would be to investigate the parameterized complexity of $d$-$SC1P$-$1E/01E$ problems on general matrices.
\label{references}

\bibliography{final_paper}

\end{document}